%% file: main.tex
\newif\ifsoda\sodafalse
\newif\ifpazo\pazofalse

\ifsoda
    \input{soda-header}

\else
    \documentclass[11pt]{article}
    \usepackage[margin=2.5cm]{geometry}
    \usepackage{times,mathptmx,bbm,microtype}
    \usepackage{amsmath,amsthm,amssymb,stmaryrd,enumitem,xspace,comment,booktabs}
    \usepackage[usenames, dvipsnames]{xcolor}
    \usepackage{algorithm}
    \usepackage[noend]{algpseudocode}
    \usepackage{thmtools,thm-restate}
    \usepackage{url}
    \usepackage[sort]{cite}
    \usepackage[colorlinks, linkcolor=BrickRed, citecolor=blue]{hyperref}
    \usepackage{cleveref}

    \DeclareMathAlphabet{\mathcal}{OMS}{cmsy}{m}{n}
    \DeclareMathAlphabet{\mathbfcal}{OMS}{cmsy}{b}{n}

    \title{Vertex Ordering Problems in Directed Graph Streams}
    \author{Amit Chakrabarti \and Prantar Ghosh \and Andrew McGregor \and Sofya Vorotnikova}
    \date{}
    \input{macros}

\begin{document}

    \maketitle
    \thispagestyle{empty}
    \input{0-abstract}
    \newpage
    \setcounter{page}{1}
\fi

\input{1-intro}
\input{2-general}
\input{2-randomgeneral}

\input{3-tournament}
\input{6-sink}
\input{4-random}

\input{5-rank}

\section*{Acknowledgements}

We thank Riko Jacob for a helpful discussion about the sink finding problem.
We thank the anonymous SODA~2020 reviewers for several helpful comments that improved
the presentation of the paper.

\bibliographystyle{abbrv}
\bibliography{refs}     

\appendix

\end{document}

%% file: macros.tex
\def\clap#1{\hbox to 0pt{\hss#1\hss}}

\newcommand{\xqedhere}[2]{\rlap{\hbox to#1{\hfil\llap{\ensuremath{#2}}}}}

\newcommand{\eps}{\varepsilon}

\newcommand{\indic}{\mathbbm{1}}
\ifsoda \renewcommand{\indic}{\mathds{1}} \fi
\ifpazo \renewcommand{\indic}{\mathbb{1}} \fi
\newcommand{\EE}{\mathbb{E}}

\newcommand{\RR}{\mathbb{R}}

\newcommand{\cA}{\mathcal{A}}

\newcommand{\cD}{\mathcal{D}}
\newcommand{\cE}{\mathcal{E}}

\newcommand{\cO}{\mathcal{O}}
\newcommand{\cQ}{\mathcal{Q}}
\newcommand{\cS}{\mathcal{S}}

\newcommand{\bx}{\mathbf{x}}
\newcommand{\by}{\mathbf{y}}
\newcommand{\bz}{\mathbf{z}}
\newcommand{\tO}{\widetilde{O}}

\DeclareMathOperator{\R}{R}
\newcommand{\best}{\beta}
\newcommand{\prb}{q}
\newcommand{\qhat}{\hat{q}}
\newcommand{\Gstar}{G^*}

\renewcommand{\ge}{\geqslant}
\renewcommand{\le}{\leqslant}
\renewcommand{\geq}{\geqslant}
\renewcommand{\leq}{\leqslant}
\renewcommand{\b}{\{0,1\}}
\newcommand{\etal}{{et al.}\xspace}

\newcommand{\Dyes}{\cD_{\mathrm{yes}}}
\newcommand{\Dno}{\cD_{\mathrm{no}}}
\newcommand{\rk}{rk}

\newcommand{\prob}[1]{\Pr\left [ #1 \right ]}
\newcommand{\ch}[2]{\ensuremath{{#1}_{#2}}}

\newcommand{\ang}[1]{\langle{#1}\rangle}
\newcommand{\ceil}[1]{\lceil{#1}\rceil}

\DeclareMathOperator{\cost}{cost}
\DeclareMathOperator{\poly}{poly}
\DeclareMathOperator{\polylog}{polylog}
\DeclareMathOperator{\err}{err}
\DeclareMathOperator{\Est}{Est}
\DeclareMathOperator{\Tou}{Tou}
\newcommand{\red}{\mathrm{red}}
\newcommand{\din}{\mathrm{d_{in}}}
\newcommand{\dout}{\mathrm{d_{out}}}

\newcommand{\Bin}{\textup{Bin}}
\newcommand{\expec}[1]{{\mathbb E}[#1]\xspace}
\newcommand{\idx}{\textsc{index}\xspace}
\newcommand{\disj}{\textsc{disj}\xspace}
\newcommand{\tpj}{\textsc{tpj}\xspace}

\newcommand{\fas}{\textsc{fas}\xspace}
\newcommand{\fast}{\textsc{fas-t}\xspace}
\newcommand{\fassiz}{\textsc{fas-size}\xspace}
\newcommand{\fassizt}{\textsc{fas-size-t}\xspace}
\newcommand{\topo}{\textsc{topo-sort}\xspace}
\newcommand{\topot}{\textsc{topo-sort-t}\xspace}
\newcommand{\acyc}{\textsc{acyc}\xspace}
\newcommand{\acyct}{\textsc{acyc-t}\xspace}
\newcommand{\sink}{\textsc{sink-find}\xspace}
\newcommand{\sinkt}{\textsc{sink-find-t}\xspace}
\newcommand{\stconn}{\textsc{stconn-dag}\xspace}
\newcommand{\shortpath}{\textsc{shortpath-dag}\xspace}
\newcommand{\sci}{\textsc{sci}\xspace}

\newcommand{\ormpjeq}{\textsc{mpj-meet}\xspace}
\newcommand{\rank}{\textsc{rank-aggr}\xspace}

\newcommand{\PlantDag}{\mathsf{PlantDAG}\xspace}

\newcommand{\pparagraph}[1]{\medskip\noindent{\bfseries #1.}~}
\ifsoda
    \crefname{fact}{fact}{facts}
    \crefname{observation}{observation}{observations}
    \crefname{@theorem}{theorem}{theorems}
    
    \newcommand{\sodaqed}{\qquad\vbox{\hrule height0.6pt\hbox{%
      \vrule height1.3ex width0.6pt\hskip0.8ex
      \vrule width0.6pt}\hrule height0.6pt}}
    \newcommand{\qed}{\sodaqed}
    \newcommand{\qedhere}{\rlap{\hbox to\textwidth{\hfil\llap{\ensuremath{\sodaqed}}}}}
    
\else

\newtheorem{theorem}{Theorem}[section]
\newtheorem{lemma}[theorem]{Lemma}

\newtheorem{corollary}[theorem]{Corollary}

\newtheorem{observation}[theorem]{Observation}
\newtheorem{claim}[theorem]{Claim}
\newtheorem{fact}[theorem]{Fact}
\newtheorem{proposition}[theorem]{Proposition}

\theoremstyle{definition}  \newtheorem{definition}[theorem]{Definition}

\fi

%% file: 0-abstract.tex
%
  We consider {\em directed} graph algorithms in a streaming setting, focusing on problems concerning orderings of the vertices. This includes such fundamental problems as topological sorting and acyclicity testing. We also study the related problems of finding a minimum feedback arc set (edges whose removal yields an acyclic graph), and finding a sink vertex. We are interested in both adversarially-ordered and randomly-ordered streams.
  For arbitrary input graphs with edges ordered adversarially, we show that most of these problems have high space complexity, precluding sublinear-space solutions. Some lower bounds also apply when the stream is randomly ordered: e.g., in our most technical result we show that testing acyclicity in the $p$-pass random-order model requires roughly $n^{1+1/p}$ space. For other problems, random ordering can make a dramatic difference: e.g., it is possible to find a sink in an acyclic tournament in the one-pass random-order model using polylog$(n)$ space whereas under adversarial ordering roughly $n^{1/p}$ space is necessary and sufficient given $\Theta(p)$ passes. We also design sublinear algorithms for the feedback arc set problem in tournament graphs; for random graphs; and for randomly ordered streams. In some cases, we give lower bounds establishing that our algorithms are essentially space-optimal. Together, our results complement the much maturer body of work on algorithms for {\em undirected} graph streams. 

%% file: 1-intro.tex

\section{Introduction} \label{sec:intro}

While there has been a large body of work on undirected graphs in the data stream model \cite{McGregor14}, the complexity of processing directed graphs (digraphs) in this model is relatively unexplored. The handful of exceptions include multipass algorithms emulating random walks in directed graphs~\cite{SarmaGP11,Jin19}, establishing prohibitive space lower bounds on finding sinks~\cite{HRR98} and answering reachability queries~\cite{FKMSZ05b}, and ruling out semi-streaming constant-pass algorithms for directed reachability~\cite{GuruswamiOnak16}. This is rather unfortunate given that many of the massive graphs often mentioned in the context of motivating work on graph streaming are directed, e.g., hyperlinks, citations, and Twitter ``follows'' all correspond to directed edges.

In this paper we consider the complexity of a variety of fundamental problems related to vertex ordering in directed graphs. For example, one basic problem that motivated\footnote{The problem was explicitly raised in an open problems session at the Shonan Workshop ``Processing Big Data Streams" (June 5-8, 2017) and generated considerable discussion.} much of this work is as follows: given a stream consisting of edges of an acyclic graph in an arbitrary order, how much memory is required to return a topological ordering of the graph? In the offline setting, this can be computed in $O(m+n)$ time using Kahn's algorithm \cite{Kahn:1962} or via depth-first trees \cite{Tarjan1976} but nothing was known in the data stream setting.

We also consider the related minimum feedback arc set problem, i.e., estimating the minimum number of edges (arcs) that need to be removed such that the resulting graph is acyclic. This problem is NP-hard and the best known approximation factor is $O(\log n \log\log n)$ for arbitrary graphs~\cite{Even1998}, although a PTAS is known in the case of tournaments~\cite{Kenyon-MathieuS06}. Again, nothing was known in the data stream model. In contrast, the analogous problem for undirected graphs is well understood in the data stream model. The number of edges required to make an undirected graph acyclic is $m-n+c$ where $c$ is the number of connected components. The number of connected components can be computed in $O(n\log n)$ space by constructing a spanning forest~\cite{FKMSZ05b,ahngm12}.



\pparagraph{Previous Work}
Some versions of the problems we study in this work have been considered previously in the query complexity model. For example, Huang et al.~\cite{HuangKK11} consider the ``generalized sorting problem" where $G$ is an acyclic graph with a unique topological order. The algorithm is presented with an undirected version of this graph and may query any edge to reveal its direction. The goal is to learn the topological ordering with the minimum number of queries. Huang et al.~\cite{HuangKK11} and Angelov et al.~\cite{AngelovKM08} also studied the average case complexity of various problems where the input graph is chosen from some known distribution. Ailon~\cite{NIPS2011_4428} studied the equivalent problem for feedback arc set in tournaments. Note that all these query complexity results are adaptive and do not immediately give rise to small-space data stream algorithms.

Perhaps the relative lack of progress on streaming algorithms for directed graph problems stems from their being considered ``implicitly hard'' in the literature, a point made in the recent work of Khan and Mehta~\cite{KhanM19}. Indeed, that work and the also-recent work of Elkin~\cite{Elkin17} provide the first nontrivial streaming algorithms for computing a depth-first search tree and a shortest-paths tree (respectively) in semi-streaming space, using $O(n/\polylog n)$ passes. Notably, fairly non-trivial work was needed to barely beat the trivial bound of $O(n)$ passes.

Some of our work here applies and extends the work of Guruswami and Onak~\cite{GuruswamiOnak16}, who gave the first super-linear (in $n$) space lower bounds in the streaming model for decision problems on graphs. In particular, they showed that solving reachability in $n$-vertex digraphs using $p$ passes requires $n^{1+\Omega(1/p)}/p^{O(1)}$ space. Via simple reductions, they then showed similar lower bounds for deciding whether a given (undirected) graph has a short $s$--$t$ path or a perfect matching.

\subsection{Results}\hfill
\begin{table*}[!hbt]
\renewcommand{\arraystretch}{1.2}
\begin{minipage}{\textwidth} 
\centering
\begin{tabular}{c c c c c}
\toprule
{\bf Problem}
& {\bf Passes}  
& {\bf Input Order} 
& {\bf Space Bound} 
& {\bf Notes}\\
\midrule
\acyc 
& $1$
&
& $\Theta(n^2)$
& \\
\acyc 
& $p$ 
& 
& $n^{1+\Omega(1/p)}/p^{O(1)}$
&\\
mult.~approx.~\fassiz 
& $1$
&
& $\Theta(n^2)$
& \\
mult.~approx.~\fassiz 
& $p$
&
& $n^{1+\Omega(1/p)}/p^{O(1)}$
& \\
\topo 
& $1$
&
& $\Theta(n^2)$
& \\
\topo 
& $p$ 
& 
& $n^{1+\Omega(1/p)}/(p+1)^{O(1)}$
&\\
mult.~approx.~\fas 
& $1$
&
& $\Theta(n^2)$
& \\
mult.~approx.~\fas 
& $p$
&
& $n^{1+\Omega(1/p)}/(p+1)^{O(1)}$
& \\
\stconn 
& $p$
& random
& $n^{1+\Omega(1/p)}/p^{O(1)}$
& error probability~$1/p^{\Omega(p)}$ \\
\acyc
& $p$
& random
& $n^{1+\Omega(1/p)}/p^{O(1)}$
& error probability $1/p^{\Omega(p)}$ \\
mult.~approx.~\fassiz
& $p$
& random
& $n^{1+\Omega(1/p)}/p^{O(1)}$
& error probability $1/p^{\Omega(p)}$ \\
\topo
& $p$
& random
& $n^{1+\Omega(1/p)}/(p+1)^{O(1)}$
& error probability $1/p^{\Omega(p)}$ \\
mult.~approx.~\fas
& $p$
& random
& $n^{1+\Omega(1/p)}/(p+1)^{O(1)}$
& error probability $1/p^{\Omega(p)}$ \\
$(1+\eps)$-approx.~\fast
& $1$
& 
& $\tO(\eps^{-2} n)$
& exp.~time post-processing \\
$3$-approx.~\fast
& $p$
& 
& $\tO(n^{1+1/p})$
& \\
\acyct
& $1$
& 
& $\tO(n)$
& \\
\acyct
& $p$
& 
& $\Omega(n/p)$
& \\
\sinkt
& $2p-1$
& 
& $\tO(n^{1/p})$
& \\
\sinkt
& $p$
& 
& $\Omega(n^{1/p}/p^2)$
& \\
\sinkt
& $1$
& random
& $\tO(1)$
& \\
\topo
& $1$
& random
& $\tO(n^{3/2})$
& random DAG + planted path \\
\topo
& $O(\log n)$
& 
& $\tO(n^{4/3})$
& random DAG + planted path \\
$(1+\eps)$-apx.~\rank
& $1$
& 
& $\tO(\eps^{-2} n)$
& exp.~time post-processing \\
\bottomrule
\end{tabular}

\end{minipage}
\caption{%
Summary of our algorithmic and space lower bound results. These problems are defined in \Cref{sec:prelim}. The input stream is adversarially ordered unless marked as ``random'' above. Besides the above results, we also give an oracle (query complexity) lower bound in \Cref{sec:fast-oracle-lbs}.%
}
\label{table:results}
\end{table*}

\pparagraph{Arbitrary Graphs} To set the stage, in \Cref{sec:gen-digraphs} we present a number of negative results for the case when the input digraph can be arbitrary. In particular, we show that there is no one-pass sublinear-space algorithm for such fundamental digraph problems as testing whether an input digraph is acyclic, topologically sorting it if it is, or finding its feedback arc set if it is not. These results set the stage for our later focus on specific families of graphs, where we can do much more, algorithmically.

For our lower bounds, we consider both arbitrary and random stream orderings. In \Cref{sec:gen-digraphs:arb}, we concentrate on the arbitrary ordering and show that checking whether the graph is acyclic, finding a topological ordering of a directed acyclic graph (DAG), or any multiplicative approximation of feedback arc set requires $\Omega(n^2)$ space in one pass. The lower bound extends to $n^{1+\Omega(1/p)}/p^{O(1)}$ when the number of passes is $p\geq 1$. In \Cref{sec:gen-digraphs:rand}, we show that essentially the same bound holds even when the stream is {\em randomly} ordered. This strengthening is one of our more technically involved results and it is based on generalizing a fundamental result by Guruswami and Onak \cite{GuruswamiOnak16} on $s$--$t$ connectivity in the multi-pass data stream model.

As a by-product of our generalization, we also obtain the first random-order super-linear (in $n$) lower bounds for the {\em undirected} graph problems of deciding (i)~whether there exists a short $s$--$t$ path (ii)~whether there exists a perfect matching.

\pparagraph{Tournaments} A {\em tournament} is a digraph that has exactly one directed edge between each pair of distinct vertices. If we assume that the input graph is a tournament, it is trivial to find a topological ordering, given that one exists, by considering the in-degrees of the vertices. Furthermore, it is known that ordering the vertices by in-degree yields a 5-approximation to feedback arc set~\cite{CoppersmithFR10}.

In \Cref{sec:tournaments}, we present an algorithm which computes a $(1+\eps)$-approximation to feedback arc set in one pass using $\tO(\eps^{-2} n)$ space\footnote{Throughout the paper, $\tO(f(n))=O(f(n) \polylog n )$.}. However, in the post-processing step, it estimates the number of back edges for every permutation of vertices in the graph, thus resulting in exponential post-processing time. Despite its ``brute force'' feel, our algorithm is essentially optimal, both in its space usage (unconditionally) and its post-processing time (in a sense we shall make precise later). We address these issues in \Cref{sec:fast-oracle-lbs}. On the other hand, in \Cref{sec:fast-multi-pass}, we show that with $O(\log n)$ additional passes it is possible to compute a 3-approximation to feedback arc set while using only polynomial time and $\tO(n)$ space.

Lastly, in \Cref{sec:sink}, we consider the problem of finding a sink in a tournament which is guaranteed to be acyclic. Obviously, this problem can be solved in a single pass using $O(n)$ space by maintaining an ``is-sink'' flag for each vertex. Our results show that for arbitrary order streams this is tight. We prove that finding a sink in $p$ passes requires $\Omega(n^{1/p}/p^2)$ space. We also provide an $O(n^{1/p} \log (3p))$-space sink-finding algorithm that uses 
$O(p)$ passes, for any $1\leq p \leq \log n$. In contrast, we show that if the stream is randomly ordered, then using $\polylog n$ space and a single pass is sufficient. This is a significant separation between the arbitrary-order and random-order data stream models.


\pparagraph{Random Graphs} In \Cref{sec:randgraphs}, we consider a natural family of random acyclic graphs (see \Cref{def:plantdag} below) and present two algorithms for finding a topological ordering of vertices. We show that, for this family, $\widetilde{O}(n^{4/3})$ space is sufficient to find the best ordering given $O(\log n)$ passes. Alternatively, $\widetilde{O}(n^{3/2})$ space is sufficient given only a single pass, on the assumption that the edges in the stream are randomly ordered.

\pparagraph{Rank Aggregation} In \Cref{sec:rank-aggr}, we consider the problem of rank aggregation (formally defined in the next section), which is closely related to the feedback arc set problem. We present a one-pass, $\widetilde{O}(\eps^{-2} n)$ space algorithm that returns $(1+\eps)$-approximation to the rank aggregation problem. The algorithm is very similar to our $(1+\eps)$-approximation of feedback arc set in tournaments and has the same drawback of using exponential post-processing time.

\bigskip

A summary of these results is given in \Cref{table:results}.

\subsection{Models and Preliminaries}\hfill
\label{sec:prelim}

\pparagraph{Vertex Ordering Problems in Digraphs} An {\em ordering} of an $n$-vertex digraph $G = (V,E)$ is a list consisting of its vertices. We shall view each ordering $\sigma$ as a function $\sigma \colon V \to [n]$, with $\sigma(v)$ being the position of $v$ in the list. To each ordering $\sigma$, there corresponds a set of {\em back edges} $B_G(\sigma) = \{(v,u) \in E:\, \sigma(u) < \sigma(v)\}$. We say that $\sigma$ is a {\em topological ordering} if $B_G(\sigma) = \varnothing$; such $\sigma$ exists iff $G$ is acyclic. We define $\best_G = \min\{|B_G(\sigma)|:\, \sigma$ is an ordering of $G\}$, i.e., the size of a minimum feedback arc set for $G$.

We now define the many interrelated digraph problems studied in this work. In each of these problems, the input is a digraph $G$, presented as a stream of its edges. The ordering of the edges is adversarial unless specified otherwise. 

\begin{description}[font=\normalfont,topsep=2pt,itemsep=0pt,labelindent=\parindent]
  \item[\acyc:] Decide whether or not $G$ is acyclic.
  \item[\topo:] Under the promise that $G$ is acyclic, output a topological ordering of its vertices.
  \item[\stconn:] Under the promise that $G$ is acyclic, decide whether it has an $s$-to-$t$ path, 
  these being two prespecified vertices. 
  \item[\sink:] Under the promise that $G$ is acyclic, output a sink of $G$.
  \item[\fassiz ($\alpha$-approximation):] Output an integer $\hat\best \in [\best_G, \alpha \best_G]$.
  \item[\fas ($\alpha$-approximation):] Output an ordering $\sigma$ such that $|B_G(\sigma)| \le \alpha \best_G$.
  \item[\fast:] Solve \fas under the promise that $G$ is a tournament. In a similar vein, we define the promise problems \acyct, \topot, \sinkt, \fassizt.
\end{description}
For randomized solutions to these problems we shall require that the error probability be at most $1/3$.

We remark that the most common definition of the minimum feedback arc set problem in the literature on optimization is to identify a small set of edges whose removal makes the graph acyclic, so \fassiz is closer in spirit to this problem than \fas. As we shall see, our algorithms will apply to both variants of the problem. On the other hand, lower bounds sometimes require different proofs for the two variants. Since $\best_G = 0$ iff $G$ is acyclic, we have the following basic observation.

\begin{observation} \label{obs:fas}
  Producing a multiplicative approximation for any of \fas, \fast, \fassiz, and \fassizt entails solving (respectively) \topo, \topot, \acyc, and \acyct.
\end{observation}

For an ordering $\pi$ of a vertex set $V$, define $E^\pi = \{(u,v) \in V^2:\, \pi(u) < \pi(v)\}$. Define $\Tou(\pi) = (V,E^\pi)$ to be the unique acyclic tournament on $V$ consistent with $\pi$.

As mentioned above, we will also consider vertex ordering problems on random graphs from a natural distribution. This distribution, which we shall call a ``planted path distribution,'' was considered by Huang et al.~\cite{HuangKK11} for average case analysis in their work on generalized sorting.

\begin{definition}[Planted Path Distribution] \label{def:plantdag} \label{def:randdag}
Let $\PlantDag_{n,\prb}$ be the distribution on digraphs on $[n]$ defined as follows. Pick a permutation $\pi$ of $[n]$ uniformly at random. Retain each edge $(u,v)$ in $\Tou(\pi)$ with probability $1$ if $\pi(v) = \pi(u)+1$, and with probability $\prb$, independently, otherwise.
\end{definition}

\pparagraph{Rank Aggregation} The feedback arc set problem in tournaments is closely related to the problem of \emph{rank aggregation} (\rank). Given $k$ total orderings $\sigma_1, \ldots, \sigma_k$ of $n$ objects we want to find an ordering that best describes the ``preferences'' expressed in the input. Formally, we want to find an ordering that minimizes 
$\cost(\pi) := \sum_{i=1}^k d(\pi, \sigma_i)$,
where the distance $d(\pi,\sigma)$ between two orderings is the number of pairs of objects ranked differently by them. That is,
\[
  d(\pi, \sigma) := \sum_{a,b\in [n]} \indic\{\pi(a) < \pi(b),\, \sigma(b) < \sigma(a)\} \,,
\]
where the notation $\indic\{\phi\}$ denotes a $0/1$-valued indicator for the condition $\phi$.

In the streaming model, the input to \rank can be given either as a concatenation of $k$ orderings, leading to a stream of length $kn$, or as a sequence of triples $(a,b,i)$ conveying that $\sigma_i(a) < \sigma_i(b)$, leading to a stream of length $k \binom{n}{2}$. Since we want the length of the stream to be polynomial in $n$, we assume $k = n^{O(1)}$.

\pparagraph{Lower Bounds through Communication Complexity} Space lower bounds for data streaming algorithms are most often proven via reductions from standard problems in communication complexity. We recall two such problems, each involving two players, Alice and Bob. In the $\idx_N$ problem, Alice holds a vector $\bx \in \b^N$ and Bob holds an index $k \in [N]$: the goal is for Alice to send Bob a message allowing him to output $x_k$. In the $\disj_N$ problem, Alice holds $\bx \in \b^N$ and Bob holds $\by \in \b^N$: the goal is for them to communicate interactively, following which they must decide whether $\bx$ and $\by$ are disjoint, when considered as subsets of $[N]$, i.e., they must output $\neg \bigvee_{i=1}^N x_i \wedge y_i$. In the special case $\disj_{N,s}$, it is promised that the cardinalities $|\bx| = |\by| = s$. In each case, the communication protocol may be randomized, erring with probability at most $\delta$. We shall use the following well-known lower bounds.

\begin{fact}[See, e.g., \cite{Ablayev96,Razborov92}] \label{fact:comm-lbs}
  For error probability $\delta = \frac13$, the one-way randomized complexity $\R_{1/3}^\to(\idx_N) = \Omega(N)$ and the general randomized complexity $\R_{1/3}(\disj_{N,N/3}) = \Omega(N)$.
\end{fact}

\pparagraph{Other Notation and Terminology} We call an edge \emph{critical} if it lies on a directed Hamiltonian path of length $n-1$ in a directed acyclic graph. We say an event holds \emph{with high probability} (w.h.p.) if the probability is at least $1-1/\poly(n)$. Given a graph with a unique total ordering, we say a vertex $u$ has rank $r$ if it occurs in the $r$th position in this total ordering. 

%% file: 2-general.tex

\section{General Digraphs and the Hardness of some Basic Problems} \label{sec:gen-digraphs}

In this section, our focus is bad news. In particular, we show that there is no one-pass sublinear-space algorithm for the rather basic problem of testing whether an input digraph is acyclic, nor for topologically sorting it if it is. These results set the stage for our later focus on {\em tournament} graphs, where we can do much more, algorithmically.

\subsection{Arbitrary Order Lower Bounds}
\label{sec:gen-digraphs:arb}
\ifsoda \mbox{}\smallskip 

\noindent
\fi
To begin, note that the complexity of \topot is easily understood: maintaining in-degrees of all vertices and then sorting by in-degree provides a one-pass $O(n\log n)$-space solution. However, the problem becomes maximally hard without the promise of a tournament.

\begin{theorem} \label{thm:topo-sort-lb}
  Solving \topo in one pass requires $\Omega(n^2)$ space.
\end{theorem}
\begin{proof}
  We reduce from $\idx_N$, where $N = p^2$ for a positive integer $p$. Using a canonical bijection from $[p]^2$ to $[N]$, we rewrite Alice's input vector as a matrix $\bx = (\bx_{ij})_{i,j\in [p]}$ and Bob's input index as $(y,z) \in [p]^2$. Our reduction creates a graph $G = (V,E)$ on $n = 4p$ vertices: the vertex set $V = L^0 \uplus R^0 \uplus L^1 \uplus R^1$, where each $|L^b| = |R^b| = p$. These vertices are labeled, with $\ell^0_i$ being the $i$th vertex in $L^0$ (and similarly for $r^0_i, \ell^1_i, r^1_i$).
  
  Based on their inputs, Alice and Bob create streams of edges by listing the following sets:
  \ifsoda
  \begin{align*}
    E_\bx &= \{(\ell^b_i, r^b_j) :\, b \in \b,\, i, j \in [p],\, x_{ij} = b\} \,, \\
    E_{yz} &= \{(r^0_z, \ell^1_y), \, (r^1_z, \ell^0_y)\} \,.
  \end{align*}
  \else
  \[
    E_\bx = \{(\ell^b_i, r^b_j) :\, b \in \b,\, i, j \in [p],\, x_{ij} = b\} \,, \quad
    E_{yz} = \{(r^0_z, \ell^1_y), \, (r^1_z, \ell^0_y)\} \,.
  \]
  \fi
  The combined stream defines the graph $G$, where $E = E_\bx \cup E_{yz}$.
  
  We claim that $G$ is acyclic. In the digraph $(V, E_\bx)$, every vertex is either a source or a sink. So the only vertices that could lie on a cycle in $G$ are $\ell^0_y, r^0_z, \ell^1_y$, and $r^1_z$. Either $(\ell^0_y, r^0_z) \notin E$ or $(\ell^1_y, r^1_z) \notin E$, so there is in fact no cycle even among these four vertices.
  
  Let $\sigma$ be a topological ordering of $G$. If $x_{yz} = 0$, then we must, in particular, have $\sigma(\ell^0_y) < \sigma(\ell^1_y)$, else we must have $\sigma(\ell^1_y) < \sigma(\ell^0_y)$. Thus, by simulating a one-pass algorithm $\cA$ on Alice's stream followed by Bob's stream, consulting the ordering $\sigma$ produced by $\cA$ and outputting $0$ iff $\sigma(\ell^0_y) < \sigma(\ell^1_y)$, the players can solve $\idx_N$. It follows that the space used by $\cA$ must be at least $\R^\to_{1/3}(\idx_N) = \Omega(N) = \Omega(p^2) = \Omega(n^2)$.
\end{proof}

For our next two results, we use reductions from \stconn. It is a simple exercise to show that a one-pass streaming algorithm for \stconn requires $\Omega(n^2)$ space. Guruswami and Onak~\cite{GuruswamiOnak16} showed that a $p$-pass algorithm requires ${n^{1+\Omega(1/p)}}/{p^{O(1)}}$ space.\footnote{Although their paper states the lower bound for $s$-$t$ connectivity in general digraphs, their proof in fact shows the stronger result that the bound holds even when restricted to DAGs.} 

\begin{proposition} \label{thm:acyc-lb}
  Solving \acyc requires $\Omega(n^2)$ space in one pass and ${n^{1+\Omega(1/p)}}/{p^{O(1)}}$ space in $p$ passes.
\end{proposition}
\begin{proof}
  Given a DAG $G$ and specific vertices $s,t$, let $G'$ be obtained by adding edge $(t,s)$ to $G$. Then $G'$ is acyclic iff $G$ has no $s$-to-$t$ path. By the discussion above, the lower bounds on \acyc follow.
\end{proof}

\begin{corollary}\label{cor:topo-multipass-lb}
Solving \topo in $p$ passes requires $n^{1+\Omega(1/p)}/(p+1)^{O(1)}$ space.   
\end{corollary}

\begin{proof}
Given a $p$-pass $S$-space algorithm $A$ for \topo, we can obtain a $(p+1)$-pass $(S+O(n\log n))$-space algorithm for \acyc as follows. Run algorithm $A$, store the ordering it outputs, and in another pass, check if the ordering induces any back-edge. If it does, we output NO, and otherwise, we output YES. In case of any runtime error, we return NO. For correctness, observe that if the input graph $G$ is acyclic, then $A$ returns a valid topological ordering w.h.p.. Hence, the final pass doesn't detect any back-edge, and we correctly output YES. In case $G$ is not acyclic, the promise that the input graph for \topo would be a DAG is violated, and hence, $A$ either raises an error or returns some ordering that must induce a back-edge since $G$ doesn't have a topological ordering. Thus, we correctly return NO in this case. Finally, \Cref{thm:acyc-lb} implies that $S+O(n\log n)\geq n^{1+\Omega(1/p)}/(p+1)^{O(1)}$, i.e., $S \geq n^{1+\Omega(1/p)}/(p+1)^{O(1)}$. 
\end{proof}

\begin{corollary} \label{cor:fas-lb}
  A multiplicative approximation algorithm for either \fassiz or \fas requires $\Omega(n^2)$ space in one pass. In $p$ passes, such approximations require ${n^{1+\Omega(1/p)}}/{p^{O(1)}}$ space and ${n^{1+\Omega(1/p)}}/{(p+1)^{O(1)}}$ space respectively.
\end{corollary}
\begin{proof}
  This is immediate from \Cref{obs:fas}, \Cref{thm:topo-sort-lb}, \Cref{thm:acyc-lb}, and \Cref{cor:topo-multipass-lb}.
\end{proof}



  
  
  

%% file: 2-randomgeneral.tex

\subsection{Random Order Lower Bounds}
\label{sec:gen-digraphs:rand}
\ifsoda \mbox{}\smallskip 

\noindent
\fi
We consider the \stconn, \acyc, and \fas problems in a uniformly randomly ordered digraph stream. Recall that for adversarially ordered streams, these problems require about $n^{1+\Omega(1/p)}$ space in $p$ passes. The hardness ultimately stems from a similar lower bound for the \shortpath problem. In this latter problem, the input is an $n$-vertex DAG with two designated vertices $v_s$ and $v_t$, such that either (a) there exists a path of length at most $2p+2$ from $v_s$ to $v_t$ or (b) $v_t$ is unreachable from $v_s$. The goal is to determine which of these is the case.

Our goal in this section is to show that the same lower bound continues to hold under random ordering, provided we insist on a sufficiently small error probability, about $1/p^{\Omega(p)}$. We prove this for \shortpath.  As this is a special case of \stconn, a lower bound for \shortpath carries over to \stconn. Further, by the reductions in \Cref{thm:acyc-lb} and \Cref{cor:fas-lb,cor:topo-multipass-lb}, the lower bounds also carry over to \acyc, \topo, and \fas. We also show a barrier result arguing that this restriction to low error is necessary: for the \shortpath problem, if an error probability of at least $2/p!$ is allowed, then $\tO(n)$ space is achievable in $p$ passes.

Our proof uses the machinery of the Guruswami--Onak lower bound for \shortpath under an adversarial stream ordering~\cite{GuruswamiOnak16}. As in their work, we derive our space lower bound from a communication lower bound for {\em set chasing intersection} (henceforth, \sci). However, unlike them, we need to prove a ``robust'' lower bound for \sci, in the sense of Chakrabarti, Cormode, and McGregor~\cite{ChakrabartiCM16}, as explained below. To define \sci, we first set up a special family of multilayer pointer jumping problems, described next. 

Picture a layered digraph $\Gstar$ with $2k+1$ layers of vertices, each layer having $m$ vertices, laid out in a rectangular grid with each column being one layer. From left to right, the layers are numbered $-k, -k+1, \ldots, k$. Layer $0$ is called the {\em mid-layer}. The only possible edges of $\Gstar$ are from layer $\ell$ to layer $\ell-1$, or from layer $-\ell$ to layer $-\ell+1$, for $\ell\in[k]$ (i.e., edges travel from the left and right ends of the rectangular grid towards the mid-layer). We designate the first vertex in layer $-k$ as $v_s$ and the first vertex in layer $k$ as $v_t$. 

Each vertex not in the mid-layer has exactly $t$ outgoing edges, numbered $1$st through $t$th, possibly with repetition (i.e., $\Gstar$ is a multigraph). Think of these edges as {\em pointers}. An input to one of our communication problems (to be defined soon) specifies the destinations of these pointers. Thus, an input consists of $2mkt$ {\em tokens}, where each token is an integer in $[m]$ specifying which of the $m$ possibilities a certain pointer takes. The pointers emanating from layer~$\ell$ of vertices constitute the $\ell$th layer of pointers. Our communication games will involve $2k$ players named $P_{-k}, \ldots, P_{-1}, P_1, \ldots, P_k$. We say that $P_\ell$ is the {\em natural owner} of the portion of the input specifying the $\ell$th layer of pointers.

In the $\sci_{m,k,t}$ problem, the goal is to determine whether or not there exists a mid-layer vertex reachable from $v_s$ as well as $v_t$. Consider the communication game where each pointer is known to its natural owner and the players must communicate in $k-1$ rounds, where in each round they broadcast messages in the fixed order $P_{-1}, \ldots, P_{-k}, P_1, \ldots, P_k$. Guruswami and Onak showed that this problem requires total communication $\Omega(m^{1+1/(2k)}/k^{16}\log^{3/2}m)$ in the parameter regime $t^{2k} \ll m$. This almost immediately implies a similar lower bound for \shortpath---simply reverse the directions of the pointers in positive-numbered layers---which then translates into a data streaming lower bound along standard lines.

The key twist in {\em our} version of the \sci problem is that each pointer is allocated to one of the $2k$ players {\em uniformly at random}: thus, most pointers are not allocated to their natural owners. The players have to determine the output to \sci communicating exactly in the same pattern as before, up to a small error probability taken over the protocol's internal randomness as well as the random allocation. This setup potentially makes the problem easier because there is a good chance that the players will be able to ``jump two pointers'' within a single round. Our main technical result is to show that a lower bound of the form $m^{1+\Omega(1/k)}$ holds despite this. In the terminology of Chakrabarti et al.~\cite{ChakrabartiCM16}, who lower-bounded a number of communication problems under such random-allocation setups, this is a {\em robust} communication lower bound.

\begin{theorem} \label{thm:sci}
  Suppose that $t^{2k} = o(m/\polylog(m))$ and that protocol $\Pi$ solves $\sci_{m,k,t}$ with error $\eps < (2k)^{-2k-2}$ when the input is randomly allocated amongst the $2k$ players, as described above. Then, $\Pi$ communicates $\Omega(m^{1+1/(2k)}/k^{16}\log^{3/2}m)$ bits.
\end{theorem}

To prove this result, we consider a problem we call $\ormpjeq_{m,k,t}$, defined next (Guruswami and Onak called this problem $\textsc{or} \circ \textsc{lpce}$). Consider an input $\Gstar$ to $\sci_{m,k,t}$ and fix an $i\in[t]$. If we retain only the $i$th pointer emanating from each vertex, for each layer $\ell$, the $\ell$th layer of pointers defines a function $f_{\ell,i} \colon [n] \to [n]$. Let $x_i$ (respectively, $y_i$) denote the index of the unique mid-layered vertex reached from $v_s$ (respectively, $v_t$) by following the retained pointers. Formally,
\ifsoda
\begin{align*}
  x_i &= f_{-1,i}( f_{-2,i}( \cdots f_{-k,i}(1) \cdots )) \,, \\
  y_i &= f_{1,i} ( f_{2,i} ( \cdots f_{k,i}(1)  \cdots )) \,.
\end{align*}
\else
\[
  x_i = f_{-1,i}( f_{-2,i}( \cdots f_{-k,i}(1) \cdots )) \,, \quad
  y_i = f_{1,i} ( f_{2,i} ( \cdots f_{k,i}(1)  \cdots )) \,.
\]
\fi
Define a function to be {\em $r$-thin} if every element in its range has at most $r$ distinct pre-images. The instance $\Gstar$ is said to {\em meet at $i$} if $x_i = y_i$ and is said to be {\em $r$-thin at $i$} if each function $f_{\ell,i}$ is $r$-thin. The desired output of $\ormpjeq$ is
\ifsoda
\begin{align*}
  \ormpjeq(\Gstar) = \bigvee_{i=1}^t &\indic\{\Gstar \text{ meets at } i\}
  \vee\\ 
  &\indic\{\Gstar \text{ is not $(C\log m)$-thin at $i$}\} \,,
\end{align*}
\else
\[
  \ormpjeq(\Gstar) = \bigvee_{i=1}^t \indic\{\Gstar \text{ meets at } i\}
  \vee \indic\{\Gstar \text{ is not $(C\log m)$-thin at $i$}\} \,,
\]
\fi
for an appropriate universal constant $C$. The corresponding communication game allocates each pointer to its natural owner and asks them to determine the output using the same communication pattern as for \sci. Here is the key result about this problem.

\begin{lemma}[Lemma~7 of Guruswami--Onak~\cite{GuruswamiOnak16}] \label{lem:ormpjeq}
  \ifsoda
  The $(k-1)$-round constant-error communication complexity of $\ormpjeq$ is lower-bounded as follows:
  \[ \R^{k-1}(\ormpjeq_{m,k,t}) = \Omega(tm / (k^{16}\log m)) - O(kt^2) \,. \]
  \mbox{}\hfill\qed
  \else
  The $(k-1)$-round constant-error communication complexity
  $\R^{k-1}(\ormpjeq_{m,k,t}) = \Omega(tm / (k^{16}\log m)) - O(kt^2)$. \qed
  \fi
\end{lemma}

Using this, we prove our main technical result.
\begin{proof}[Proof of \Cref{thm:sci}]
  Based on the $\eps$-error protocol $\Pi$ for $\sci_{m,k,t}$, we design a protocol $\cQ$ for $\ormpjeq_{m,k,t}$ as follows. Let $\Gstar$ be an instance of $\ormpjeq$ allocated to players as described above. The players first check whether, for some $i$, $\Gstar$ fails to be $r$-thin at $i$, for $r := C\log m$: this check can be performed in the first round of communication with each player communicating a single bit. If the check passes, the protocol ends with output $1$. From now on, we assume that $\Gstar$ is indeed $r$-thin at each $i\in [t]$.
  
  Using public randomness, the players randomly renumber the vertices in each layer of $\Gstar$, creating an instance $G'$ of \sci.\footnote{This step is exactly as in Guruswami-Onak~\cite{GuruswamiOnak16}. Formally, each function $f_{\ell,i}$ is replaced by a corresponding function of form $\pi_{\ell,i} \circ f_{\ell,i} \circ \pi_{\ell+1,i}^{-1}$ (for $\ell > 0$), for random permutations $\pi_{\ell,i} \colon [m] \to [m]$. To keep things concise, we omit the full details here.}
  The players then choose $\rho$, a random allocation of pointers as in the \sci problem. They would like to simulate $\Pi$ on $G'$, as allocated by $\rho$, but of course they can't do so without additional communication. Instead, using further public randomness, for each pointer that $\rho$ allocates to someone besides its natural owner, the players reset that pointer to a uniformly random (and independent) value in $[m]$. We refer to such a pointer as {\em damaged}. Since there are $2k$ players, each pointer is damaged with probability $1-1/(2k)$. Let $G''$ denote the resulting random instance of \sci. The players then simulate $\Pi$ on $G''$ as allocated by $\rho$.
  
  It remains to analyze the correctness properties of $\cQ$. Suppose that $\Gstar$ is a $1$-instance of \ormpjeq. Then there exists $i\in[t]$ such that $\Gstar$ meets at $i$. By considering the unique maximal paths out of $v_s$ and $v_t$ following only the $i$th pointers at each vertex, we see that $\Gstar$ is also a $1$-instance of \sci. Since the vertex renumbering preserves connectivity, $G'$ is also a $1$-instance of \sci. With probability $(2k)^{-2k}$, none of the $2k$ pointers on these renumbered paths is damaged; when this event occurs, $G''$ is also a $1$-instance of \sci. Therefore, $\cQ$ outputs $1$ with probability at least $(2k)^{-2k}(1-\err(\Pi)) \ge (2k)^{-2k}(1-\eps)$.
  
  Next, suppose that $\Gstar$ is a $0$-instance of \ormpjeq. It could be that $\Gstar$ is a $1$-instance of \sci. However, as Guruswami and Onak show,\footnote{%
  See the final paragraph of the proof of Lemma~11 in~\cite{GuruswamiOnak16}.}
  the random vertex renumbering ensures that $\Pr[\sci(G') = 1] < o(1)$. For the rest of the argument, assume that $\sci(G') = 0$. In order to have $\sci(G'') = 1$, there must exist a mid-layer vertex $x$ such that
  \ifsoda
  \begin{align} 
    x
    &= f_{1,i_1}( f_{2,i_2}( \cdots f_{k,i_k}(1) \cdots )) \notag \\
    &= f_{-1,j_1}( f_{-2,j_2}( \cdots f_{-k,j_k}(1) \cdots ))
    \label{eq:intersect}
  \end{align}
  \else
  \begin{equation} \label{eq:intersect}
    f_{1,i_1}( f_{2,i_2}( \cdots f_{k,i_k}(1) \cdots )) = x
    = f_{-1,j_1}( f_{-2,j_2}( \cdots f_{-k,j_k}(1) \cdots ))
  \end{equation}
  \fi
for some choice of pointer numbers $i_1, \ldots, i_k,$ $j_1, \ldots, j_k \in [t]$. We consider three cases.

  \begin{itemize}
  \item {\em Case 1: None of the pointers in the above list is damaged.~}
  In this case, \cref{eq:intersect} cannot hold, because $\sci(G') = 0$.
  
  \item {\em Case 2: The layer-$1$ pointer in the above list is damaged.~}
  Condition on a particular realization of pointers in negative-numbered layers and let $x$ denote the mid-layered vertex reached from $v_s$ by following pointers numbered $j_k,\ldots,j_1$, as in \cref{eq:intersect}. The probability that the damaged pointer at layer~$1$ points to $x$ is $1/m$. Since this holds for each conditioning, the probability that $\sci(G'') = 1$ is also $1/m$.
  
  \item {\em Case 3: The layer-$\ell$ pointer is damaged, but pointers in layers $1,\ldots,\ell-1$ are not, where $\ell \ge 2$.~} Again, condition on a particular realization of pointers in negative-numbered layers and let $x$ be as above. Since the functions $f$ in \cref{eq:intersect} are all $r$-thin, the number of vertices in layer $\ell-1$ that can reach $x$ using only undamaged pointers is at most $r^{\ell-1} \le r^{k-1}$. The probability that the damaged pointer at layer~$\ell$ points to one of these vertices is at most $r^{k-1}/m$. 
  \end{itemize}
  
  Combining the cases, the probability that \cref{eq:intersect} holds for a particular choice of pointer numbers $i_1, \ldots, i_k, j_1, \ldots, j_k \in [t]$ is at most $r^{k-1}/m$. Taking a union bound over the $t^{2k}$ choices, the overall probability $\Pr[\sci(G'') = 1] < t^{2k} r^{k-1}/m = o(1)$, for the parameter regime $t^{2k} = o(m/\polylog(m))$ and $r = O(\log m)$. Therefore, $\cQ$ outputs $1$ with probability at most $\err(\Pi) + o(1) \le \eps + o(1)$.
  
  Thus far, we have a protocol $\cQ$ that outputs $1$ with probability $\alpha$ when $\ormpjeq(\Gstar) = 0$ and with probability $\beta$ when $\ormpjeq(\Gstar) = 1$, where $\alpha \le \eps + o(1)$ and $\beta \ge (2k)^{-2k}(1-\eps)$. Recall that $\eps = (2k)^{-2k-2}$, so $\beta$ is bounded away from $\alpha$. Let $\cQ'$ be a protocol where we first toss an unbiased coin: if it lands heads, we output $0$ with probability $\delta := (\alpha+\beta)/2$ and $1$ with probability $1-\delta$; if it lands tails, we simulate $\cQ$. Then $\cQ'$ is a protocol for \ormpjeq with error probability $\frac12 - (\beta-\alpha)/4$. By \Cref{lem:ormpjeq}, this protocol must communicate $\Omega(m^{1+1/(2k)}/k^{16}\log^{3/2}m)$ bits and so must $\Pi$.  
\end{proof}

By a standard reduction from random-allocation communication protocols to random-order streaming algorithms, we obtain the following lower bound: the main result of this section.

\begin{theorem} \label{thm:rand-order-shortpath}
  For each constant $p$, a $p$-pass algorithm that solves \shortpath on $n$-vertex digraphs whose edges presented in a uniform random order, erring with probability at most $1/p^{\Omega(p)}$ must use $\Omega(n^{1+1/(2p+2)}/\log^{3/2}n)$ bits of space. 
  
  Consequently, similar lower bounds hold for the problems \stconn, \acyc, \topo, \fas, and \fassiz.
  \qed
\end{theorem}

This paper is focused on {\em directed} graph problems. However, it is worth noting that a by-product of our generalization of the Guruswami--Onak bound to randomly ordered streams is that we also obtain the first random-order super-linear (in $n$) lower bounds for two important {\em undirected} graph problems.
\begin{corollary} \label{cor:rand-order-undir}
  For each constant $p$, $n^{1+\Omega(1/p)}$ space is required to solve either of the following problems in $p$ passes, erring with probability at most $1/p^{\Omega(p)}$, over a randomly ordered edge stream of an $n$-vertex undirected graph $G$:
  \begin{itemize}
    \item decide whether $G$ contains a perfect matching;
    \item decide whether the distance between prespecified vertices $v_s$ and $v_t$ is at most $2p+2$.
  \end{itemize}
\end{corollary}

\pparagraph{A Barrier Result}
Notably, \Cref{thm:rand-order-shortpath} applies only to algorithms with a rather small error probability. This is inherent: allowing just a slightly larger error probability renders the problem solvable in semi-streaming space. This is shown in the result below, which should be read as a {\em barrier} result rather than a compelling algorithm.

\begin{proposition} \label{thm:shortpath-rand-order}
  Given a randomly ordered edge stream of a digraph $G$, the \shortpath problem on $G$ can be solved using $\tO(n)$ space and $p$ passes, with error probability at most $2/p!$\,.
\end{proposition}

\begin{proof}
  Recall that we're trying to decide whether or not $G$ has a path of length at most $(2p+2)$ from $v_s$ to $v_t$. The high-level idea is that thanks to the random ordering, a ``Bellman--Ford'' style algorithm that grows a forward path out of $v_s$ and a backward path out of $v_t$ is very likely to make more than one step of progress during {\em some} pass.
  
  To be precise, we maintain arrays $d_s$ and $d_t$, each indexed by $V$. Initialize the arrays to $\infty$, except that $d_s[v_s] = d_t[v_t] = 0$. During each pass, we use the following logic.
  \begin{tabbing}
  xxx\=xxx\=\kill
    \> \small for each edge $(x,y)$ in the stream: \\
    \> \> \small if $d_s[x] + d_t[y] \le 2p+1$: output {\sc true} and halt \\
    \> \> \small $d_s[y] \gets \min( d_s[y], 1+d_s[x] )$ \\
    \> \> \small $d_t[x] \gets \min( d_t[x], 1+d_t[y] )$
  \end{tabbing}
  If we complete $p$ passes without any output, then we output {\sc false}.
  
  If $G$ has no short enough path from $v_s$ to $v_t$, this algorithm will always output {\sc false}. So let's consider the other case, when there {\em is} a $v_s$--$v_t$ path $\pi$ of length at most $2p+2$. Let vertex $z$ be the midpoint of $\pi$, breaking ties arbitrarily if needed. The subpaths $[v_s,z]_\pi$ and $[z,v_t]_\pi$ have lengths $q$ and $r$, respectively, with $q \le p+1$ and $r \le p+1$. Notice that if our algorithm is allowed to run for $q$ (resp.~$r$) passes, then $d_s[z]$ (resp.~$d_t[z]$) will settle to its correct value. If both of them settle, then the algorithm correctly outputs {\sc true}. So, the only nontrivial case is when $q,r \in \{p,p+1\}$.
  
  Let $E_s$ be the event that the random ordering of the edges in the stream places the edges of $[v_s,z]_\pi$ in the exact reverse order of $\pi$. Let $E_t$ be the event that the random ordering places the edges of $[z,v_t]_\pi$ in the exact same order as $\pi$. If $E_s$ does not occur, then for some two consecutive edges $(w,x), (x,y)$ on $[v_s,z]_\pi$, the stream puts $(w,x)$ before $(x,y)$. Therefore, once $d_s[w]$ settles to its correct value, the following pass will settle not just $d_s[x]$, but also $d_s[y]$; therefore, after $q-1 \le p$ passes, $d_s[z]$ is settled. Similarly, if $E_t$ does not occur, then after $r-1 \le p$ passes, $d_t[z]$ is settled. As noted above, if both of them settle, the algorithm correctly outputs {\sc true}.
  
  Thus, the error probability $\le \Pr[E_s \vee E_t] \le \Pr[E_s] + \Pr[E_t] = 1/q! + 1/r! \le 2/p!$\,, as required.
\end{proof}

%% file: 3-tournament.tex

\section{Feedback Arc Set in Tournaments} \label{sec:tournaments}


\subsection{Accurate, One Pass, but Slow Algorithm for FAS-T}
\label{sec:fast-one-pass}
\ifsoda \mbox{}\smallskip 

\noindent
\fi
We shall now design an algorithm for \fast (that also solves \fassizt) based on linear sketches for $\ell_1$-norm estimation. Recall that the $\ell_1$-norm of a vector $\bx \in \RR^N$ is $\|\bx\|_1 = \sum_{i\in [N]} |x_i|$. A $d$-dimensional $\ell_1$-sketch with accuracy parameter $\eps$ and error parameter $\delta$ is a distribution $\cS$ over $d\times N$ matrices, together with an estimation procedure $\Est \colon \RR^d \to \RR$ such that
\[
  \Pr_{S \gets \cS} \big[\,
  (1-\eps) \|\bx\|_1 \le \Est(S\bx) \le (1+\eps) \|\bx_1\| 
  \,\big] \ge 1-\delta \,.
\]
Such a sketch is ``stream friendly'' if there is an efficient procedure to generate a given column of $S$ and further, $\Est$ is efficient. Obviously, a stream friendly sketch leads to a space and time efficient algorithm for estimating $\|\bx\|_1$ given a stream of entrywise updates to $\bx$. We shall use the following specialization of a result of Kane \etal~\cite{kanenw10}.

\begin{fact}[Kane \etal~\cite{kanenw10}] \label{fact:l1}
  There is a stream friendly $d$-dimensional $\ell_1$-sketch with accuracy $\eps$ and error $\delta$ that can handle $N^{O(1)}$ many $\pm1$-updates to $\bx \in \RR^N$, with each update taking $O(\eps^{-2} \log\eps^{-1} \log\delta^{-1} \log N)$ time, with $d = O(\eps^{-2} \log\delta^{-1})$, and with entries of the sketched vector fitting in $O(\log N)$ bits.
\end{fact}


\begin{theorem} \label{thm:fast-exp}
  There is a one-pass algorithm for \fast that uses $O(\eps^{-2} n \log^2 n)$ space and
  returns a $(1+\eps)$-approximation with probability at least $\frac23$, but requires exponential post-processing time.
\end{theorem}
\begin{proof}
  Identify the vertex set of the input graph $G=(V,E)$ with $[n]$ and put $N = \binom{n}{2}$. We index vectors $\bz$ in $\RR^N$ as $z_{uv}$, where $1 \le u < v \le n$. Define a vector $\bx \in \b^N$ based on $G$ and vectors $\by^\pi \in \b^N$ for each permutation $\pi \colon [n] \to [n]$ using indicator variables as follows.
  \[
    x_{uv} = \indic\{(u,v) \in E\} \,, \quad
    y^\pi_{uv} = \indic\{\pi(u) < \pi(v)\} \,.
  \]
  A key observation is that the $uv$-entry of $\bx - \by^\pi$ is nonzero iff the edge between $u$ and $v$ is a back edge of $G$ according to the ordering $\pi$. Thus, $|B_G(\pi)| = \|\bx - \by^\pi\|_1$. 
  
  Our algorithm processes the graph stream by maintaining an $\ell_1$-sketch $S\bx$ with accuracy $\eps/3$ and error $\delta = 1/(3\cdot n!)$. By \Cref{fact:l1}, this takes $O(\eps^{-2} n \log^2 n)$ space and $O(\eps^{-2} \log\eps^{-1} n \log^2 n)$ time per edge.
  
  In post-processing, the algorithm considers all $n!$ permutations $\pi$ and, for each of them, computes $S(\bx - \by^\pi) = S\bx - S\by^\pi$. It thereby recovers an estimate for $\|\bx - \by^\pi\|_1$ and finally outputs the ordering $\pi$ that minimizes this estimate. By a union bound, the probability that every estimate is $(1\pm\eps/3)$-accurate is at least $1 - n!\cdot\delta = 2/3$. When this happens, the output ordering provides a $(1+\eps)$-approximation to \fast by our key observation above.
\end{proof}

Despite its ``brute force'' feel, the above algorithm is essentially optimal, both in its space usage (unconditionally) and its post-processing time (in a sense we shall make precise later). We address these issues in \Cref{sec:fast-oracle-lbs}.

\subsection{Multiple Passes: FAS-T in Polynomial Time}
\label{sec:fast-multi-pass}
\ifsoda \mbox{}\smallskip 

\noindent
\fi
For a more time-efficient streaming algorithm, we design one based on the  \textsc{KwikSort} algorithm of Ailon et al.~\cite{AilonCN08}. This (non-streaming) algorithm operates as follows on a tournament $G=(V,E)$.
\begin{itemize}[itemsep=1pt]
    \item Choose a random ordering of the vertices: $v_1, v_2, \ldots , v_n$.
    \item Vertex $v_1$ partitions $V$ into two sub-problems $\{u:\, (u,v_1)\in E\}$ and $\{w:\, (v_1,w)\in E\}$. 
    At this point we know the exact place of $v_1$ in the ordering.
    \item Vertex $v_2$ further partitions one of the these sub-problems. Proceeding in this manner, after $v_1, v_2, \ldots, v_i$ are considered, there are $i+1$ sub-problems.
    \item Continue until all $n$ vertices are ordered.
\end{itemize}
When $v_i$ is being used to divide a sub-problem we refer to it as a \emph{pivot}.

\pparagraph{Emulating \textsc{KwikSort} in the Data Stream Model}
We will emulate \textsc{KwikSort} in $p$ passes over the data stream. In each pass, we will consider the action of multiple pivots. Partition $v_1, \ldots , v_n$ into $p$ groups $V_1, \ldots, V_p$, where $V_1 = \{v_1, \ldots, v_{cn^{1/p}\log n}\}$, $V_2$ consists of the next $cn^{2/p}\log n$ vertices in the sequence, and $V_j$ contains $cn^{j/p} \log n$ vertices coming after $V_{j-1}$. Here $c$ is a sufficiently large constant. 
At the end of pass $j+1$, we want to emulate the effect of pivots in $V_{j+1}$ on the sub-problems resulting from considering pivots in $V_1$ through $V_{j}$. In order to do that, in pass $j+1$ for each vertex $v \in V_{j+1}$ we store all edges between $v$ and vertices in the same sub-problem as $v$, where the sub-problems are defined at the end of pass $j$.

The following combinatorial lemma plays a key role in  analyzing this algorithm's space usage.

\begin{lemma}[Mediocrity Lemma] \label{lemma:division}
In an $n$-vertex tournament, if we pick a vertex $v$ uniformly at random, then \mbox{$\Pr[\eps n < \din(v) < (1-\eps)n] \ge 1-4\eps$}. Similarly, \mbox{$\Pr[\eps n < \dout(v) < (1-\eps)n] \ge 1-4\eps$}. In particular, with probability at least $1/3$, $v$ has in/out-degree between $n/6$ and $5n/6$.\footnote{
The Mediocrity Lemma is tight: consider sets of vertices $A, B,C$ where $|A|=|C|=2\eps n$ and $|B|=(1-4\eps)n$. Edges on $B$ do not form any directed cycles. Subgraphs induced by $A$ and $C$ are balanced, i.e., the in-degree equals the out-degree of every vertex (where degrees here are considered within the subgraph). All other edges are directed from $A$ to $B$, from $B$ to $C$, or from $A$ to $C$. Then vertices with in/out-degrees between $\eps n$ and $(1-\eps)n$ are exactly the vertices in $B$, and a random vertex belongs to this set with probability $1-4\eps$.
}
\end{lemma}
\begin{proof}
Let $H$ be a set of vertices of in-degree at least $(1-\eps)n$. Let $h = |H|$. On the one hand, $\sum_{v \in H} \din(v) \ge (1-\eps)nh$. On the other hand, the edges that contribute to the in-degrees of vertices in $H$ have both endpoints in $H$ or one endpoint in $H$ and one in $V \setminus H$. The number of such edges is
\[
  \sum_{v \in H} \din(v) \le \binom{h}{2} + h(n-h)
  = \frac{1}{2}(2nh - h^2 - h) \,.
\]
Therefore, $(2nh - h^2 - h)/2 \ge (1-\eps)nh$. This implies $h < 2\eps n$.

Thus, the number of vertices with in-degree at least $(1-\eps)n$ (and out-degree at most $\eps n$) is 
$h < 2\eps n$.
By symmetry, the number of vertices with out-degree at least $(1-\eps)n$ (and in-degree at most $\eps n$) is also less than $2\eps n$. Thus, the probability a random vertex has in/out-degree between $\eps n$ and $(1-\eps)n$ is
$(n-2h)/{n} > (n-2 \cdot 2\eps n)/{n} = 1-4\eps$.
\end{proof}

\pparagraph{Space Analysis}
Let $M_{j}$ be the maximum size of a sub-problem after pass $j$.
The number of edges collected in pass $j+1$ is then at most $M_{j} |V_{j+1}|$. By \Cref{lem:space} (below), this is at most $cn^{1+1/p}\log n$. Once the post-processing of pass $j+1$ is done, the edges collected in that pass can be discarded. 

\begin{lemma}\label{lem:space}
With high probability,  $M_{j}\leq n^{1-j/p}$ for all $j$.
\end{lemma}
\begin{proof}
  Let $M_j^v$ denote the size of the sub-problem that contains $v$, after the $j$th pass. We shall prove that, for each $v$, $\Pr[M_j^v > n^{1-j/p}] \le 1/n^{10}$. The lemma will then follow by a union bound.
  
  Take a particular vertex $v$. If, before the $j$th pass, we already have $M_{j-1}^v \le n^{1-j/p}$, there is nothing to prove. So assume that $M_{j-1}^v > n^{1-j/p}$.
  Call a pivot ``good'' if it reduces the size of the sub-problem containing $v$ by a factor of at least $5/6$.
  A random pivot falls in the same sub-problem as $v$ with probability at least $n^{1-j/p}/n$; when this happens, by the Mediocrity Lemma, the  probability that the pivot is good is at least $1/3$. Overall, the probability that the pivot is good is at least $n^{-j/p}/3$.

  In the $j$th pass, we use $cn^{j/p}\log n$ pivots. If at least $\log_{6/5} n$ of them are good, we definitely have $M_j^v \le n^{1-j/p}$. Thus, by a Chernoff bound, for a sufficiently large $c$, we have
  \ifsoda
  \begin{align*}
    \Pr&\left[M^v_{j} > n^{1-j/p}\right] \\
    &\leq \Pr\left[\Bin\left(cn^{j/p}\log n,\, n^{-j/p}/3\right) < \log_{6/5} n\right] \\
    &\leq 1/n^{10} \,.
  \end{align*}
  \vskip-\lastskip\hfill
  \else
  \[
    \Pr\left[M^v_{j} > n^{1-j/p}\right]
    \leq \Pr\left[\Bin\left(cn^{j/p}\log n,\, n^{-j/p}/3\right) < \log_{6/5} n\right]
    \leq 1/n^{10} \,.
    \qedhere
  \]
  \fi
\end{proof}

\begin{theorem}
There exists a polynomial time $p$-pass data stream algorithm using $\tO(n^{1+1/p})$ space that returns a $3$-approximation (in expectation) for \fast.
\end{theorem}
\begin{proof}
The pass/space tradeoff follows from Lemma~\ref{lem:space} and the discussion above it; the approximation factor follows directly from the analysis of Ailon et al. \cite{AilonCN08}.
\end{proof}

\subsection{A Space Lower Bound}
\label{sec:acyct-space-lb}
\ifsoda \mbox{}\smallskip 

\noindent
\fi
Both our one-pass algorithm and the $O(\log n)$-pass instantiation of our multi-pass algorithm use at least $\Omega(n)$ space. For \fassizt, where the desired output is a just a number, it is reasonable to ask whether $o(n)$-space solutions exist. We now prove that they do not.

\begin{proposition} \label{thm:acyc-t-lb}
  Solving \acyct is possible in one pass and $O(n\log n)$ space. Meanwhile, any $p$-pass solution requires $\Omega(n/p)$ space.
\end{proposition}
\begin{proof}
  For the upper bound, we maintain the in-degrees of all vertices in the input graph $G$. Since $G$ is a tournament, the set of in-degrees is exactly $\{0,1,\ldots,n-1\}$ iff the input graph is acyclic.
  
  For the lower bound, we reduce from $\disj_{N,N/3}$. Alice and Bob construct a tournament $T$ on $n = 7N/3$ vertices, where the vertices are labeled $\{v_1, \ldots, v_{2N}, w_1, \ldots, w_{N/3}\}$. Alice, based on her input $\bx$, adds edges $(v_{2i}, v_{2i-1})$ for each $i \in \bx$. For each remaining pair $(i,j) \in [2N]\times [2N]$ with $i<j$, she adds the edge $(v_i,v_j)$.
  Let $a_1 < \cdots < a_{N/3}$ be the sorted order of the elements in Bob's set $\by$. For each $k = a_\ell \in \by$, Bob defines the alias $v_{2N+k} = w_\ell$ and then adds the edges
  \ifsoda
  \begin{align*}
    E_k = \,&\{(v_i, v_{2N+k}):\, 1 \le i \le 2k-1\}
    ~\cup \\
    & \{(v_{2N+k}, v_j):\, 2k \le j \le 2N\} \,.
  \end{align*}
  \else
  \[
    E_k = \{(v_i, v_{2N+k}):\, 1 \le i \le 2k-1\}
    \cup \{(v_{2N+k}, v_j):\, 2k \le j \le 2N\} \,.
  \]
  \fi
Finally, he adds the edges $\{(w_i, w_j):\, 1 \le i < j \le N/3\}$. This completes the construction of $T$.

We claim that the tournament $T$ is acyclic iff $\bx \cap \by = \varnothing$. The ``only if'' part is direct from construction, since if $\bx$ and $\by$ intersect at some index $k \in [N]$, we have the directed cycle $(v_{2k},v_{2k-1},v_{2N+k},v_{2k})$.
For the ``if'' part, let $\sigma$ be the ordering $(v_1,\ldots,v_{2N})$ and let $T' = \Tou(\sigma)$, as defined in \Cref{sec:prelim}. We show how to modify $\sigma$ into a topological ordering of $T$, proving that $T$ is acyclic. Observe that, by construction, the tournament $T\setminus  \{w_1,\ldots,w_{N/3}\}$ can be obtained from $T'$ by flipping only the edges $(v_{2i-1},v_{2i})$ for each $i \in \bx$. Each time we perform such an edge flip, we modify the topological ordering of $T'$ by swapping the associated vertices of the edge. The resultant ordering would still be topological as the vertices were consecutive in the ordering before the flip. Thus, after performing these swaps, we get a topological ordering of  $T\setminus \{w_1,\ldots,w_{N/3}\}$. Now, consider some $k \in \by$. Since $\bx \cap \by = \varnothing$, $k \notin \bx$ and so, $v_{2k}$ succeeds $v_{2k-1}$ in this ordering, just as in $\sigma$, since we never touched these two vertices while performing the swaps. Thus, for each such $k$, we can now insert $v_{2N+k}$ between $v_{2k-1}$ and $v_{2k}$ in the ordering and obtain a topological ordering of $T$. This proves the claim.

Thus, given a $p$-pass solution to \acyct using $s$ bits of space, we obtain a protocol for $\disj_{N,N/3}$ that communicates at most $(2p-1)s$ bits. By \Cref{fact:comm-lbs}, $(2p-1)s = \Omega(N) = \Omega(n)$, i.e., $s = \Omega(n/p)$.
\end{proof}

\begin{theorem} \label{cor:fast-space-lb}
A $p$-pass multiplicative approximation for \fassizt requires $\Omega(n/p)$ space.
\end{theorem}
\begin{proof}
  This is immediate from \Cref{obs:fas,thm:acyc-t-lb}.
\end{proof}

\subsection{An Oracle Lower Bound}
\label{sec:fast-oracle-lbs}
\ifsoda \mbox{}\smallskip 

\noindent
\fi
Let us now consider the nature of the post-processing performed by our one-pass \fast algorithm in \Cref{sec:fast-one-pass}. During its streaming pass, that algorithm builds an {\em oracle} based on $G$ that, when queried on an ordering $\sigma$, returns a fairly accurate estimate of $|B_G(\sigma)|$. It proceeds to query this oracle $n!$ times to find a good ordering. This raises the question: is there a more efficient way to exploit the oracle that the algorithm has built? A similar question was asked in Bateni et al.~\cite{BateniEM17} in the context of using sketches for the maximum coverage problem.

Were the oracle {\em exact}---i.e., on input $\sigma$ it returned $|B_G(\sigma)|$ exactly---then two queries to the oracle would determine which of $(i,j)$ and $(j,i)$ was an edge in $G$. It follows that $O(n\log n)$ queries to such an exact oracle suffice to solve \fast and \fassizt. However, what we actually have is an {\em $\eps$-oracle}, defined as one that, on query $\sigma$, returns $\hat\best \in \RR$ such that $(1-\eps)|B_G(\sigma)| \le \hat\best \le (1+\eps)|B_G(\sigma)|$. We shall show that an $\eps$-oracle cannot be exploited efficiently: a randomized algorithm will, with high probability, need exponentially many queries to such an oracle to solve either \fast or \fassizt.

To prove this formally, we consider two distributions on $n$-vertex tournaments, defined next.


\begin{definition} \label{def:tou-dists}
  Let $\Dyes, \Dno$ be distributions on tournaments on $[n]$ produced as follows. To produce a sample from $\Dyes$, pick a permutation $\pi$ of $[n]$ uniformly at random; output $\Tou(\pi)$. To produce a sample from $\Dno$, for each $i,j$ with $1 \le i < j \le n$, independently at random, include edge $(i,j)$ with probability $\frac12$; otherwise include edge $(j,i)$.
\end{definition}

Let $\sigma$ be an ordering of $[n]$. By linearity of expectation, if $T$ is sampled from either $\Dyes$ or $\Dno$,
\[
  \EE |B_T(\sigma)| = m := \frac12 \binom{n}{2} \,.
\]
In fact, we can say much more.
\begin{lemma} \label{lem:back-edge-conc}
  There is a constant $c$ such that, for all $\eps > 0$, sufficiently large $n$, a fixed ordering $\sigma$ on $[n]$, and random $T$ drawn from either $\Dyes$ or $\Dno$,
  \[
    \Pr\left[ (1-\eps)m < |B_T(\sigma)| < (1+\eps)m \right] \ge 1 - 2^{-c\eps^2 n} \,.
  \]
\end{lemma}
\begin{proof}
  When $T \gets \Dno$, the random variable $|B_T(\sigma)|$ has binomial distribution $\Bin(2m,\frac12)$, so the claimed bound is immediate.
  
  Let $T \gets \Dyes$. Partition the edges of the tournament into perfect matchings $M_1, \ldots, M_{n-1}$. 
  For each $i\in[n-1]$, let $X_i$ be the number of back edges of $T$ involving $M_i$, i.e.,
  \ifsoda
  \begin{align*}
    X_i = |\{(u,v)\in M_i:\,& \text{ either } (u,v) \in B_T(\sigma)\\
    &\text{ or } (v,u) \in B_T(\sigma)\}| \,.
  \end{align*}
  \else
  \[
    X_i = |\{(u,v)\in M_i:\, \text{either } (u,v) \in B_T(\sigma) \text{ or } (v,u) \in B_T(\sigma)\}| \,.
  \]
  \fi
  Notice that $X_i \sim \Bin(n/2,\frac12)$, whence 
  \[
    \Pr\left[ \textstyle (1-\eps)n/4 < X_i < \frac12 (1+\eps)n/4 \right] \ge 1 - 2^{b\eps^2 n} \,,
  \]
  for a certain constant $b$. By a union bound, the probability that {\em all} of the $X_i$s are between these bounds is at least $1 - (n-1)2^{-b\eps^2 n} \ge 1 - 2^{-c\eps^2 n}$, for suitable $c$. When this latter event happens, we also have $(1-\eps)m < |B_T(\sigma)| =
  \frac12\sum_{i=1}^{n-1} X_i < (1+\eps)m$.
\end{proof}



We define a {\em $(q,\eps)$-query algorithm} for a problem $P$ to be one that access an input digraph $G$ solely through queries to an $\eps$-oracle and, after at most $q$ such queries, outputs its answer to $P(G)$. We require this answer to be correct with probability at least $\frac23$.

Now consider the particular oracle $\cO_{T,\eps}$, describing an $n$-vertex tournament $T$, that behaves as follows when queried on an ordering $\sigma$.
\begin{itemize}[itemsep=0pt]
  \item If $(1-\eps/2)m < |B_T(\sigma)| < (1+\eps/2)m$, then return $m$.
  \item Otherwise, return $|B_T(\sigma)|$.
\end{itemize}
Clearly, $\cO_{T,\eps}$ is an $\eps$-oracle. The intuition in the next two proofs is that this oracle makes life difficult by seldom providing useful information.

\begin{proposition} \label{thm:oracle-topo}
  Every $(q,\eps)$-query algorithm for \topot makes $\exp(\Omega(\eps^2 n))$ queries.
\end{proposition}
\begin{proof}
WLOG, consider a $(q,\eps)$-query algorithm, $\cA$, that makes exactly $q$ queries, the last of which is its output. Using Yao's minimax principle, fix $\cA$'s random coins, obtaining a deterministic $(q,\eps)$-query algorithm $\cA'$ that succeeds with probability $\ge \frac23$ on a random tournament $T \gets \Dyes$. Let $\sigma_1, \ldots, \sigma_q$ be the sequence of queries that $\cA'$ makes when the answer it receives from the oracle to each of $\sigma_1, \ldots, \sigma_{q-1}$ is $m$.
  
  Suppose that the oracle supplied to $\cA'$ is $\cO_{T,\eps}$. Let $\cE$ be the event that $\cA'$'s query sequence is $\sigma_1, \ldots, \sigma_q$ and it receives the response $m$ to each of these queries. For a particular $\sigma_i$,
  \ifsoda
  \begin{align*}
    \Pr&[\cO_{T,\eps}(\sigma_i) = m] \\
    &= \Pr[(1-\eps/2)m < |B_T(\sigma_i)| < (1+\eps/2)m] \\
    &\ge 1 - 2^{-b\eps^2 n} 
  \end{align*}
  \else
  \[
    \Pr[\cO_{T,\eps}(\sigma_i) = m]
    = \Pr[(1-\eps/2)m < |B_T(\sigma_i)| < (1+\eps/2)m]
    \ge 1 - 2^{-b\eps^2 n} 
  \]
  \fi
  for a suitable constant $b$, by \Cref{lem:back-edge-conc}. Thus, by a union bound, $\Pr[\cE] \ge 1 - q2^{-b\eps^2 n}$.
  
  When $\cE$ occurs, $\cA'$ must output $\sigma_q$, but $\cE$ itself implies that $|B_T(\sigma_q)| \ne 0$, so $\cA'$ errs. Thus, the success probability of $\cA'$ is at most $1-\Pr[\cE] \le q2^{-b\eps^2 n}$. Since this probability must be at least $\frac23$, we need $q \ge \frac23\cdot 2^{b\eps^2 n} = \exp(\Omega(\eps^2 n))$.
\end{proof}

\begin{proposition} \label{thm:oracle-acyc}
  Every $(q,\eps)$-query algorithm for \acyct makes $\exp(\Omega(\eps^2 n))$ queries.
\end{proposition}
\begin{proof}
  We proceed similarly to \Cref{thm:oracle-topo}, except that we require the deterministic $(q,\eps)$-query algorithm $\cA'$ to succeed with probability at least $\frac23$ on a random $T \gets \frac12(\Dyes + \Dno)$. We view $T$ as being chosen in two stages: first, we pick $Z \in_R \{\mathrm{yes},\mathrm{no}\}$ uniformly at random, then we pick $T \gets \cD_Z$.
  
  Define $\sigma_1, \ldots, \sigma_q$ and $\cE$ as before. So $\Pr[\cE] \ge 1 - q2^{-b\eps^2 n}$. When $\cE$ occurs, $\cA'$ must output some fixed answer, either ``yes'' or ``no.'' We consider these cases separately.
  
  Suppose that $\cA'$ outputs ``no,'' declaring that $T$ is not acyclic. Then $\cA'$ errs whenever $Z =$ yes and $\cE$ occurs. The probability of this is at least $\frac12 - q2^{-b\eps^2 n}$, but it must be at most $\frac13$, requiring $q = \exp(\Omega(\eps^2 n))$.
  
  Suppose that $\cA'$ outputs ``yes'' instead. Then it errs when $Z =$ no, $T$ is cyclic, and $\cE$ occurs. Since
  \[
    \Pr[T \text{ acyclic} \mid Z = \text{no}]
    = n!/2^{\binom{n}{2}} = \exp(-\Omega(n^2)) \,,
  \]
  we have $\frac13 \ge \Pr[\cA'$ errs$] \ge \frac12 - \exp(-\Omega(n^2)) - q2^{-b\eps^2 n}$, requiring $q = \exp(\Omega(\eps^2 n))$.
\end{proof}

\begin{theorem}
  A $(q,\eps)$-query algorithm that gives a multiplicative approximation for either \fast or \fassizt must make $q = \exp(\Omega(\eps^2 n))$ queries.
\end{theorem}
\begin{proof}
  This is immediate from \Cref{obs:fas,thm:oracle-topo,thm:oracle-acyc}.
\end{proof}

%% file: 6-sink.tex
\section{Sink Finding in Tournaments}
\label{sec:sink}

A classical offline algorithm for \topo is to repeatedly find a sink $v$ in the input graph (which must exist in a DAG), prepend $v$ to a growing list, and recurse on $G \setminus v$. Thus, \sink itself is a fundamental digraph problem. Obviously, \sink can be solved in a single pass using $O(n)$ space by maintaining an ``is-sink'' flag for each vertex. Our results below show that for arbitrary order streams this is tight, even for tournament graphs.

In fact, we say much more. In $p$ passes, on the one hand, the space bound can be improved to roughly $O(n^{2/p})$. On the other hand, any $p$-pass algorithm requires about $\Omega(n^{1/p})$ space. While these bounds don't quite match, they reveal the correct asymptotics for the number of passes required to achieve polylogarithmic space usage: namely, $\Theta(\log n/ \log\log n)$.

In contrast, we show that if the stream is randomly ordered, then using $\polylog(n)$ space and a single pass is sufficient. This is a significant separation between the adversarial and random order data streams.




\subsection{Arbitrary Order Sink Finding}
\label{sec:sink:arb}

\begin{theorem}[Multi-pass algorithm] \label{thm:sink-t}
  For all $p$ with $1\le p \le \log n$, there is a $(2p-1)$-pass algorithm for \sinkt that uses $O(n^{1/p} \log (3p))$ space and has failure probability at most $1/3$.
\end{theorem}
\begin{proof}
  Let the input digraph be $G = (V,E)$. For a set $S \subseteq V$, let $\max S$ denote the vertex in $S$ that has maximum in-degree. This can also be seen as the maximum vertex within $S$ according to the total ordering defined by the edge directions.
  
  Our algorithm proceeds as follows.
  \begin{itemize}[itemsep=1pt]
    \item {\em Initialization:} Set $s = \ceil{n^{1/p}\ln (3p)}$. Let $S_1$ be a set of $s$ vertices chosen randomly from $V$.
    \item For $i=1$ to $p-1$:
    \begin{itemize}[topsep=2pt,itemsep=1pt]
      \item {\em During pass $2i-1$:} Find $v_{i}=\max S_{i}$ by computing the in-degree of each vertex in $S_{i}$.
      \item {\em During pass $2i$:} Let $S_{i+1}$ be a set of $s$ vertices chosen randomly from $\{u:\, (v_i,u)\in E\}$.
    \end{itemize}
    \item {\em During pass $2p-1$:} Find $v_{p}=\max S_{p}$ by computing the in-degree of each vertex in $S_p$.
  \end{itemize}

For the sake of analysis, consider the quantity $\ell_i=|\{u:(v_i,u)\in E\}|$. Note that, for each $i \in [p]$,
\[
  \prob{\ell_i > \ell_{i-1}/n^{1/p}}
  = (1-1/n^{1/p})^{s}
  \le \frac{1}{3p} \,.
\]
Thus, by the union bound, $\ell_p = 0$ with probability at least $1-p/(3p)= 2/3$.  Note that $\ell_p=0$ implies that $v_p$ is a sink.
\end{proof}

We turn to establishing a multi-pass lower bound. Our starting point for this is the {\em tree pointer jumping} problem $\tpj_{k,t}$, which is a communication game involving $k$ players. To set up the problem, consider a complete ordered $k$-level $t$-ary tree $T$; we consider its root $z$ to be at level $0$, the children of $z$ to be at level $1$, and so on. We denote the $i$-th child of $y \in V(T)$ by $\ch{y}{i}$, the $j$-th child of $\ch{y}{i}$ by $\ch{y}{i,j}$, and so on. Thus, each leaf of $T$ is of the form $\ch{z}{i_1,\ldots,i_{k-1}}$ for some integers $i_1, \ldots, i_{k-1} \in [t]$.

An instance of $\tpj_{k,t}$ is given by a function $\phi \colon V(T) \to [t]$ such that $\phi(y) \in \b$ for each leaf $y$. The desired one-bit output is
\begin{align}
\tpj_{k,t}(\phi) &:= g^{(k)}(z) = g(g(\cdots g(z) \cdots)) \,, \text{where} \notag\\
g(y) &:= \begin{cases}
  \phi(y) \,, & \text{if $y$ is a leaf,} \\
  \ch{y}{\phi(y)} \,, & \text{otherwise.}
\end{cases} \label{eq:g-def}
\end{align}
For each $j \in \{0,\ldots,k-1\}$, Player~$j$ receives the input values $\phi(y)$ for each vertex $y$ at level $j$. The players then communicate using at most $k-1$ {\em rounds}, where a single round consists of one message from each player, speaking in the order Player~$k-1$, \ldots, Player~$0$. All messages are broadcast publicly (equivalently, written on a shared blackboard) and may depend on public random coins. The cost of a round is the {\em total} number of bits communicated in that round and the cost of a protocol is the {\em maximum}, over all rounds, of the cost of a round. The randomized complexity $\R^{k-1}(\tpj_{k,t})$ is the minimum cost of a $(k-1)$-round $\frac13$-error protocol for $\tpj_{k,t}$.

Combining the lower bound approach of Chakrabarti et al.~\cite{ChakrabartiCM16} with the improved round elimination analysis of Yehudayoff~\cite{Yehudayoff16}, we obtain the following lower bound on the randomized communication complexity of the problem.
%
\begin{theorem} \label{thm:tpj}
$\R^{k-1}(\tpj_{k,t}) = \Omega(t/k)$. \qed
\end{theorem}

Based on this, we prove the following lower bound.

\begin{theorem}[Multi-pass lower bound]
  Any streaming algorithm that solves \sinkt in $p$ passes must use $\Omega(n^{1/p}/p^2)$ space.
\end{theorem}
\begin{proof}
  We reduce from $\tpj_{k,t}$, where $k = p+1$. We continue using the notations defined above. At a high level, we encode an instance of $\tpj$ in the directions of edges in a tournament digraph $G$, where $V(G)$ can be viewed as two copies of the set of leaves of $T$. Formally,
  \[
    V(G) = \{ \ang{i_1, \ldots, i_{k-1}, a} :\, \text{each } i_j \in [t] \text{ and } a \in \b \} \,.
  \]
  We assign each pair of distinct vertices $u,v \in V(G)$ to a {\em level} in $\{0, \ldots, k-1\}$ as follows. Suppose that $u = \ang{i_1, \ldots, i_k}$ and $v = \ang{i'_1, \ldots, i'_k}$. We assign $\{u,v\}$ to level $j-1$, where $j$ is the smallest index such that $i_j \ne i'_j$. Given an instance of $\tpj_{k,t}$, the players jointly create an instance of \sinkt as follows. For each $j$ from $k-1$ to $0$, in that order, Player~$j$ assigns directions for all pairs of vertices at level $j$, obtaining a set $E_j$ of directed edges, and then appends $E_j$ to a stream. The combined stream $E_{k-1} \circ \cdots \circ E_1 \circ E_0$ defines the tournament $G$. It remains to define each set $E_j$ precisely.

  The set $E_{k-1}$ encodes the bits $\phi(y)$ at the leaves $y$ of $T$ as follows.
  \begin{align} \label{eq:leaf-bits}
    \ifsoda
    E_{k-1} = \{ ( \ang{\mathbf{i}, 1-a}, \ang{\mathbf{i}, a} ) \in V(G)^2 :\, 
    \phi(\ch{z}{\mathbf{i}}) = a \} \,,
    \else
    E_{k-1} = \{ ( \ang{i_1, \ldots, i_{k-1}, 1-a}, \ang{i_1, \ldots, i_{k-1}, a} ) \in V(G)^2 :\, 
    \phi(\ch{z}{i_1, \ldots, i_{k-1}}) = a \} \,,
    \fi
  \end{align}
  Notice that if we ignore edge directions, $E_{k-1}$ is a perfect matching on $V(G)$.
  
  Now consider an arbitrary level $j \in \{0,\ldots,k-2\}$. Corresponding to each vertex $\ch{z}{i_1, \ldots, i_{j-1}}$ at level $j$ of $T$, we define the permutation $\pi_{i_1, \ldots, i_{j-1}} \colon [t] \to [t]$ thus:
  \ifsoda
  \begin{gather} \label{eq:pi-def}
    (\pi_{i_1, \ldots, i_{j-1}}(1), \ldots, \pi_{i_1, \ldots, i_{j-1}}(t)) \\
    = (1, \ldots, \ell-1, \ell+1, \ldots, t, \ell) \,, \notag\\
    \text{where } \ell = \phi( \ch{z}{i_1, \ldots, i_{j-1}} ) \notag\,.
  \end{gather}
  \else
  \begin{align} \label{eq:pi-def}
    (\pi_{i_1, \ldots, i_{j-1}}(1), \ldots, \pi_{i_1, \ldots, i_{j-1}}(t)) &= (1, \ldots, \ell-1, \ell+1, \ldots, t, \ell) \,, \notag\\
    \text{where } \ell &= \phi( \ch{z}{i_1, \ldots, i_{j-1}} ) \,.
  \end{align}
  \fi
  Using this, we define $E_j$ so as to encode the pointers at level $j$ as follows.
  \ifsoda
  \begin{align}
    E_j = \{& ( \ang{i_1, \ldots, i_{j-1}, i_j, \ldots, i_k}, \ang{i_1, \ldots, i_{j-1}, i'_j, \ldots, i'_k} ) \notag\\
    &\in V(G)^2 :\,
    \pi_{i_1, \ldots, i_{j-1}}^{-1}(i_j) < \pi_{i_1, \ldots, i_{j-1}}^{-1}(i'_j) \}  \label{eq:ej-def} \,.
  \end{align}
  \else  
  \begin{equation} \label{eq:ej-def}
    E_j = \{ ( \ang{i_1, \ldots, i_{j-1}, i_j, \ldots, i_k}, \ang{i_1, \ldots, i_{j-1}, i'_j, \ldots, i'_k} ) \in V(G)^2 :\,
    \pi_{i_1, \ldots, i_{j-1}}^{-1}(i_j) < \pi_{i_1, \ldots, i_{j-1}}^{-1}(i'_j) \} \,.
  \end{equation}
  \fi
  
  It should be clear that the digraph $(V(G), E_0 \cup E_1 \cup \cdots \cup E_{k-1})$ is a tournament. We argue that it is acyclic. Suppose, to the contrary, that $G$ has a cycle $\sigma$. Let $j \in \{0,\ldots,k-2\}$ be the smallest-numbered level of an edge on $\sigma$. Then there exist $h_1, \ldots, h_{j-1}$ such that every vertex on $\sigma$ is of the form $\ang{h_1, \ldots, h_{j-1}, i_j, \ldots, i_k}$. Let $v^{(1)}, \ldots, v^{(r)}$ be the vertices on $\sigma$ whose outgoing edges belong to level $j$. For each $q \in [r]$, let $v^{(q)} = \ang{h_1, \ldots, h_{j-1}, i^{(q)}_j, \ldots, i^{(q)}_k}$. Let $\hat\pi = \pi_{h_1, \ldots, h_{j-1}}$. According to \cref{eq:ej-def},
  \[
    \hat\pi^{-1}\big(i^{(1)}_j\big) < \hat\pi^{-1}\big(i^{(2)}_j\big) < \cdots <
    \hat\pi^{-1}\big(i^{(r)}_j\big) < \hat\pi^{-1}\big(i^{(1)}_j\big) \,,
  \]
  a contradiction.

  It follows that $G$ has a unique sink. Let $v = \ang{h_1, \ldots, h_{k-1}, a} \in V(G)$ be this sink. In particular, for each level $j \in \{0,\ldots,k-2\}$, all edges in $E_j$ involving $v$ must be directed towards $v$. According to \cref{eq:ej-def}, we must have $\pi_{h_1, \ldots, h_{j-1}}^{-1}(h_j) = t$, i.e., $\pi_{h_1, \ldots, h_{j-1}}(t) = h_j$. By \cref{eq:pi-def}, this gives $\phi(\ch{z}{h_1, \ldots, h_{j-1}}) = h_j$. Next, by \cref{eq:g-def}, this gives $g(\ch{z}{h_1, \ldots, h_{j-1}}) = \ch{z}{h_1, \ldots, h_j}$. Instantiating this observation for $j = 0, \ldots, k-2$, we have
  \ifsoda
  \begin{gather*}
    \ch{z}{h_1} = g(z),~~
    \ch{z}{h_1,h_2} = g(\ch{z}{h_1}),~~ \ldots \\
    \ldots,~~ \ch{z}{h_1,\ldots,h_{k-1}} = g(\ch{z}{h_1,\ldots,h_{k-2}}) \,,
  \end{gather*}
  \else
  \[
    \ch{z}{h_1} = g(z),~~
    \ch{z}{h_1,h_2} = g(\ch{z}{h_1}),~~ \ldots,~~
    \ch{z}{h_1,\ldots,h_{k-1}} = g(\ch{z}{h_1,\ldots,h_{k-2}}) \,,
  \]
  \fi
  i.e., $\ch{z}{h_1,\ldots,h_{k-1}} = g^{(k-1)}(z)$.

  At this point $h_1, \ldots, h_{k-1}$ have been determined, leaving only two possibilities for $v$. We now use the fact that the sole edge in $E_{k-1}$ involving $v$ must be directed towards $v$. According to \cref{eq:leaf-bits}, $\phi(\ch{z}{h_1,\ldots,h_{k-1}}) = a$. Invoking \cref{eq:g-def} again, $a = \phi(g^{(k-1)}(z)) = g^{(k)}(z) = \tpj_{k,t}(\phi)$.

  Thus, the players can read off the desired output $\tpj_{k,t}(\phi)$ from the identity of the unique sink of the constructed digraph $G$. Notice that $n := |V(G)| = 2t^{k-1}$. It follows that a $(k-1)$-pass streaming algorithm for \sinkt that uses $S$ bits of space solves $\tpj_{k,t}$ in $k-1$ rounds at a communication cost of $kS$. By \Cref{thm:tpj}, we have $S = \Omega(t/k^2) = \Omega(n^{1/(k-1)}/k^2)$.
\end{proof}

\subsection{Random Order Sink Finding}
\label{sec:sink:rand}
\ifsoda \mbox{}\smallskip 

\noindent
\fi
In this section we show that it is possible to find the sink of an acyclic tournament in one pass over a randomly order stream while using only $\polylog(n)$ space. The algorithm we consider is as follows:
  \begin{itemize}[itemsep=1pt]
    \item {\em Initialization:} Let $S$ be a random set of $s=200 \log n$ nodes.
    \item For $i=1$ to $k:=\log_2 \left (\frac{m}{200000 n\log n}\right )$:
    \begin{itemize}[topsep=2pt,itemsep=1pt]
      \item Ingest the next $c_i:=100 \cdot 2^i (n-1) \log n$ elements of the stream: For each $v\in S$, collect the set of edges $S_v$ consisting of all outgoing edges; throw away $S_v$ if it exceeds  size $220 \log n$
      \item Pick any $v\in S$, such that $|S_v|=
       (200\pm 20) \log n$ and let $S$ be the endpoints (other than $v$) of the edges in $S_v$
    \end{itemize}
    \item Ingest the next $m/1000$ elements: find $P$ the set of vertices $w$ such that there exists an edge $uw$ for some $u\in S$
    \item Ingest the remaining  $499m/500$ elements: Output any vertex in $P$ with no outgoing edges.
  \end{itemize}
  
\begin{theorem} \label{thm:sink-t-onepass}
  There is a single pass algorithm for \sinkt that uses $O(\polylog n)$ space and has failure probability at most $1/3$ under the assumption that the data stream is randomly ordered.
\end{theorem}
\begin{proof}

    We refer to the $c_i$ elements used in the iteration $i$ as the $i$th segment of the stream. For a node $u$, let $X_{u,i}$ be the number of outgoing edges from $u$ amongst the $i$th segment. The following claim follows from the Chernoff bound:
    \begin{claim}
    With high probability, for all $u$ with  $|\rk(u)-n/2^i|\geq 0.2 \cdot n/2^i$ then 
    \[|X_{u,i}-200\log n|> 0.1\cdot 200 \log n \ .\] With high probability, for all $u$ with $|\rk(u)-n/2^i|\leq 0.05 \cdot n/2^i$, then \[|X_{u,i}-200\log n|< 0.1\cdot 200 \log n \ .\]
    \end{claim}
    If follows from the claim that if after processing the $i$th segment of the stream there exists a $v$ such that $|S_v|=(200\pm 20) \log n$ then with high probability $\rk(u)=(1\pm 0.2) \cdot n/2^i$. We next need to argue that there exists such a $v$.
    
    \begin{claim}
    With high probability, for every node $u$ with $\rk(u)=(1\pm 0.2) \cdot n/2^{i-1}$, there exists an edge $uv$ in the $i$th segment such that $|\rk(v)-n/2^i|\leq 0.05 \cdot n/2^i$.
    \end{claim}
    \begin{proof}
    There are at least $0.01 \cdot n/2^i$ such edges. The probability that none of them exists in the $i$th segment is at most $(1-c_i/m)^{0.01 \cdot n/2^i}\leq 1/\poly(n)$.
    \end{proof}
    The above two claims allow us to argue by induction that we will have an element $u$ with $\rk(u)=(1\pm 0.2) \cdot n/2^i$ after the $i$th segment. At the end of the $k$th segment we have identified at least $(200-20)\log n$ vertices where every rank is at most $(1+0.2) \cdot n/2^k=O(\log n)$. With probability at least $1-1/\poly(n)$ one of these vertices includes an edge to the sink amongst the $(k+1)$ segment and hence the sink is in $P$ with high probability. There may be other vertices in $P$ but the following claim shows that we will identify any false positives while processing the final $499m/500$ elements of the stream.
    
    \begin{claim}
    With probability at least $1-1/499$, there exists at least once outgoing edge from every node except the sink amongst the last $499m/500$ elements of the stream
    \end{claim}
    \begin{proof}[Proof of Claim]
    The probability no outgoing edge from the an element of rank $r>0$ appears in the suffix of the stream is at most $(1-499/500)^r$ . Hence, by the union bound the probability that there exists an element of rank $r>0$ without an outgoing edge is at most $\sum_{r\geq 1} (1-499/500)^r=1/499.$
    \end{proof}

This concludes the proof of \Cref{thm:sink-t-onepass}.
\end{proof}

%% file: 4-random.tex

\section{Topological Ordering in Random Graphs} \label{sec:randgraphs}

%
We present results for computing a topological ordering of $G\sim \PlantDag_{n,\prb}$ (see \Cref{def:plantdag}). We first present an $O(\log n)$-pass algorithm using $\tO(n^{4/3})$ space. We then present a one-pass algorithm that uses $\tO(n^{3/2})$ space and requires the assumption that the stream is in random order.

\subsection{Arbitrary Order Algorithm}
\ifsoda \mbox{}\smallskip 

\noindent
\fi
In this section, we present two different algorithms. The first is appropriate when $\prb$ is large whereas the second is appropriate when $\prb$ is small. Combining these algorithms and considering the worst case value of $\prb$ yields the algorithm using  $\tO(n^{4/3})$ space.


\pparagraph{Algorithm for large \textit{\prb}}
The basic approach is to emulate QuickSort. We claim that we can find the relationship between any vertex $u$ among $n$ vertices and a predetermined vertex $v$ using three passes and ${O}(n+\prb^{-3} \log n)$ space. Assuming this claim, we can sort in $O(\log (\prb^2 n))$ passes and $\tO(n/\prb)$ space: we recursively partition the vertices and suppose at the end of a phase we have  sub-problems of sizes $n_1, n_2, n_3, \ldots $. Any sub-problem with at least $1/\prb^2$ vertices is then sub-divided by picking $\Theta(\log n)$ random  pivots (with replacement) within the sub-problems using the aforementioned three pass algorithm. There are at most $\prb^2 n$ such sub-problems. Hence, the total space required partition all the sub-problems in this way is at most 
\[{O}\left (\log n \sum_{i=1}^{\prb^2 n } (n_i+\prb^{-3} \log n)\right )= O(n  \prb^{-1} \log^2 n)   \ .\]
Note that the size of every sub-problem decreases by a factor at least $2$ at each step with high probability and hence after $\log (q^2 n)$ iterations, all sub-problems have at most $1/q^2$ vertices. Furthermore, each vertex degree is $O(1/\prb\cdot \log n)$ in each sub-problem. Hence, the entire remaining instance can be stored using $O(n/\prb\cdot \log n)$ space.

It remains to prove our three-pass claim. For this, we define the following families of sets:
\ifsoda
\begin{align*}
  L_i &= \{u: \exists \mbox{ $u$-to-$v$ path of length} \leq i\} \,, \\
  R_i &= \{u: \exists \mbox{ $v$-to-$u$ path of length} \leq i\} \,.
\end{align*}
\else
\[
  L_i = \{u: \exists \mbox{ $u$-to-$v$ path of length} \leq i\} \,, \quad
  R_i = \{u: \exists \mbox{ $v$-to-$u$ path of length} \leq i\} \,.
\]
\fi
Using two passes and $O(n \log n)$ space we can identify $L_2$ and $R_2$ using $O(n \log n)$ space. Let $U$ be the set of vertices not contained in $L_2\cup R_2$. The following lemma (which can be proved via Chernoff bounds) establishes that $L_2\cup R_2$ includes most of the vertices of the graph with high probability.

\begin{lemma}\label{lem:usmall}
With high probability, $|U|=O(\prb^{-2} \log n)$.
\end{lemma} 
In a third pass, we store every edge between vertices in $U$ and also compute $L_3$ and $R_3$. Computing $L_3$ and $R_3$ requires only $O(n\log n)$ space. There is an edge between each pair of vertices in $U$ with probability $\prb$ and hence, the expected number of edges between vertices in $U$ is at most $\prb|U|^2=O(\prb^{-3} \log^2 n)$. By an application of the Chernoff Bound, this bound also holds w.h.p.
Note that $L_3, R_3$, and the edges within $U$ suffice to determine whether $u\in L_\infty$ or $u\in R_\infty$ for all $u$. To see this first suppose $u\in L_\infty$ and that $(u,w)$ is the critical edge on the directed path from $u$ to $v$. Either $w \in L_2$ and therefore we deduce $u\in L_3$; or $u\in L_2$; or $u\not \in L_2$ and $w\not \in L_2$ and we therefore store the edge $(u,w)$.%

This establishes the following lemma.
\begin{lemma}\label{lem:bigp}
  There is a $O(\log n)$-pass, $\tO(n/\prb)$-space algorithm for \topo on a random input graph $G\sim \PlantDag_{n,\prb}$.
  \qed
\end{lemma}

\pparagraph{Algorithm for small \textit{\prb}}
We use only two passes. In the first pass, we compute the in-degree of every vertex. In the second, we store all edges between vertices where the in-degrees differ by at most $3\sqrt{c n \prb \cdot \ln n }$ where $c>0$ is a sufficiently large constant.

\begin{lemma}\label{lem:smallp}
  There is a two-pass, $\tO(n^{3/2} \sqrt{\prb})$-space algorithm for \topo on a random input graph $G\sim \PlantDag_{n,\prb}$.
\end{lemma}
\begin{proof}
  We show that, with high probability, the above algorithm collects all critical edges and furthermore only collects $\tO(n^{3/2} \sqrt{\prb})$ edges in total.
Let $u$ be the element of rank $r_u$. Note that 
$\din(u)$ has distribution $1+\Bin(r_u-2,\prb)$.
Let $X_u=\din(u)-1$. By an application of the Chernoff Bound,
\[
\prob{|X_u-(r_u-2)\prb|\geq \sqrt{c(r_u-2)\prb\ln n }} 
\leq 1/\poly(n) \,.
\]
Hence, w.h.p., $r_u=2+X_u/\prb \pm \sqrt{c n/\prb \cdot \ln n }$ for all vertices $u$.  Therefore, if $(u,v)$ is critical, then 
\ifsoda
\begin{align*}
  |X_u-X_v|
  &\le |X_u-(r_u-2)\prb| + \\
  &\qquad |(r_u-2)\prb-(r_v-2)\prb| + |X_v-(r_v-2)\prb| \\
  &\le 3\sqrt{c n \prb \cdot \ln n } \,.
\end{align*}
\else
\[
|X_u-X_v|  \leq  |X_u-(r_u-2)\prb|+|(r_u-2)\prb-(r_v-2)\prb| + |X_v-(r_v-2)\prb| 
 \leq   3\sqrt{c n \prb \cdot \ln n } \,.
\]
\fi
This ensures that the algorithm collects all critical edges.
For the space bound, we first observe that for an arbitrary pair of vertices $u$ and $v$, if $|X_u-X_v| \leq 3\sqrt{c n \prb \cdot \ln n }$ then   
\[
| r_u-r_v| \leq |X_u-X_v|/\prb + 2 \sqrt{c n/\prb \cdot \ln n } 
\leq 
8 \sqrt{c n/\prb \cdot \ln n }  \ . 
\]
Hence, we only store an edge between vertex $u$ and vertices whose rank differs by at most $8 \sqrt{c n/\prb \cdot \ln n }$. Since edges between such vertices are present with probability $\prb$, the expected number of edges stored incident to $u$ is  $8 \sqrt{c n \prb \cdot \ln n }$ and is $ O(\sqrt{n \prb \cdot \ln n })$ by an application of the Chernoff bounds. Across all vertices this means the number of edges stored is $ O(n^{3/2} \sqrt{\prb \cdot \ln n })$ as claimed. 
\end{proof}

Combining Lemma \ref{lem:bigp} and Lemma \ref{lem:smallp} yields the main theorem of this section.

\begin{theorem}
  There is an $O(\log n)$-pass algorithm for \topo on a random input $G\sim \PlantDag_{n,\prb}$ that uses $\tO(\min(n/\prb,n^{3/2}\sqrt{\prb})$ space. For the worst-case over $\prb$, this is $\tO(n^{4/3})$.
  \qed
\end{theorem}

\subsection{Random Order Algorithm}
\ifsoda \mbox{}\smallskip 

\noindent
\fi
The {\em transitive reduction} of a DAG $G=(V,E)$ is the minimal subgraph $G^\red=(V,E')$ such that, for all $u,v\in V$, if $G$ has a $u$-to-$v$ path, then so does $G^\red$. So if $G$ has a Hamiltonian path, $G^\red$ {\em is} this path.

The one-pass algorithm assuming a random ordering of the edges is simply to maintain $G^\red$ as $G$ is streamed in, as follows. Let $S$ be initially empty. For each edge $(u,v)$ in the stream, we add $(u,v)$ to $S$ and then remove all edges $(u',v')$ where there is a $u'$-to-$v'$ path among the stored edges.

\begin{theorem}
\ifsoda
There is a one-pass algorithm that uses $\tO(\max_{\qhat\leq \prb} \min \{n/\qhat, n^2 \qhat\})$ space and solves \topo on an input $G\sim \PlantDag_{n,\prb}$ presented in random order. In the worst case this space bound is $O(n^{3/2})$.
\else
There is a one-pass $\tO(\max_{\qhat\leq \prb} \min \{n/\qhat, n^2 \qhat\})$-space algorithm for \topo on inputs $G\sim \PlantDag_{n,\prb}$ presented in random order. In the worst case this space bound is $\tO(n^{3/2})$.
\fi
\end{theorem}

\begin{proof}
Consider the length-$T$ prefix of the stream where the edges of $G$ are
presented in random order. It will be convenient to write $T = n^2 \qhat$. We will
argue that the number of edges in the transitive reduction of this prefix is
$O(\min\{n/\qhat,\, n^2 \qhat\})$ with high probability; note the bound $n^2 \qhat$ follows trivially because the transitive reduction has at most $T$ edges. The result then follows by taking the maximum over all prefixes.

We say an edge $(u,v)$ of $G$ is \emph{short} if the difference between the
ranks is
$
  r_v-r_u \le \tau := c \qhat^{-2} \log n$ 
where $c$ is some sufficiently large constant. An edge that is not short is defined to be \emph{long}. Let $S$  be the number of short edges  in $G$ and let $M$ be the total number
of edges in $G$. Note that $\expec{S}\leq (n-1)+\prb\tau n$ and $\expec{M}=(n-1)+\prb\binom{n-1}{2}$. By the Chernoff bound, $S\leq 2\prb \tau n$ and $n^2 \prb/4 \leq M\leq n^2 \prb$ with high probability. Furthermore, the number of short edges in the prefix is expected to be $T\cdot S/M$ and, with high probability, is at most 
\[
  2T\cdot S/M \le \frac{4T \prb\tau n}{n^2 \prb/4}
  = 16cn/\qhat \cdot \log n\,.
\]
Now consider how many long edges are in the transitive reduction of the prefix. For any long edge $(u,v)$, let $X_{w}$ denote the event that $(u,w),(w,v)$ are both in the prefix. Note that the variables $\{X_w\}_{w:r_u+1\leq r_w\leq r_v-1}$ are negatively correlated and that
\[
  \prob{X_w=1} \ge (\prb T/M)^2/2 \ge  \qhat^2/2 \,.
\]

Hence, if $X=\sum_{w:r_u+1\leq r_w\leq r_v-1} X_w$ then 
\[
  \expec{X} \ge c \qhat^{-2} \log n \cdot  \qhat^2/2
  = c/2 \cdot \log n
\]
and so, by the Chernoff bound, $X>0$ with high probability and if this is the case, even if $(u,v)$ is in the prefix, it will not be in the transitive reduction of the prefix. Hence, by the union bound, with high probability no long edges exist in the transitive closure of the prefix.
\end{proof}

%% file: 5-rank.tex
\section{Rank Aggregation}
\label{sec:rank-aggr}

Recall the \rank problem and the distance $d$ between permutations, defined in \Cref{sec:prelim}. To recap,
the distance between two orderings is the number of pairs of objects which are ranked differently by them, i.e.,
\[
  d(\pi, \sigma) := \sum_{a,b\in [n]} \indic\{ \pi(a) < \pi(b),\, \sigma(b) < \sigma(a)\} \,.
\]
Note that \rank is equivalent to finding the median of a set of $k$ points under this distance function, which can be shown to be metric. It follows that picking a random ordering from the $k$ input orderings provides a $2$-approximation for \rank. 

A different approach is to reduce \rank to the {\em weighted} feedback arc set problem on a tournament. This idea leads to a $(1+\eps)$-approximation via $\ell_1$-norm estimation in a way similar to the algorithm in Section~\ref{sec:fast-one-pass}.
%
 Define a vector $\bx$ of length $\binom{n}{2}$ indexed by pairs of vertices $\{a,b\}$ where 
 \[ 
 x_{a,b} = \sum_{i=1}^k \indic\{\sigma_i(a) < \sigma_i(b)\} \,, 
 \]
i.e., the number of input orderings that have $a<b$. Then for any ordering $\pi$ define a vector $\by^\pi$, where for each pair of vertices $\{a,b\}$,
\[ 
 y^\pi_{a,b} = k \cdot \indic\{\pi(a) < \pi(b)\} \,.
\]
It is easy to see that 
$
\|\bx-\by^\pi\|_1 = \cost(\pi)
$.

As in Section~\ref{sec:fast-one-pass}, our algorithm maintains an $\ell_1$-sketch $S\bx$ with accuracy $\eps/3$ and error $\delta = 1/(3\cdot n!)$. By \Cref{fact:l1}, this requires at most $O(\eps^{-2} n \log^2 n)$ space. In post-processing, the algorithm considers all $n!$ permutations $\pi$ and, for each of them, computes $S(\bx - \by^\pi) = S\bx - S\by^\pi$. It thereby recovers an estimate for $\|\bx - \by^\pi\|_1$ and finally outputs the ordering $\pi$ that minimizes this estimate. 

The analysis of this algorithm is essentially the same as  in Theorem~\ref{thm:fast-exp}. Overall, we obtain the following result.


\begin{theorem} \label{thm:rank-aggr-restated}
  There is a one-pass algorithm for rank aggregation that uses $O(\eps^{-2} n \log^2 n)$ space, returns a $(1+\eps)$-approximation with probability at least $\frac23$, but requires exponential post-processing time.
  \qed
\end{theorem}